\documentclass[11pt]{amsart}

\usepackage{amsmath,amssymb,amsthm,amsfonts,mathrsfs}
\usepackage[authoryear,round]{natbib}
\usepackage[english]{babel}
\usepackage[utf8]{inputenc}      
\usepackage{graphicx}
\usepackage{pdfpages}
\usepackage{setspace}
\usepackage{booktabs}
\usepackage{tabularx}
\usepackage{algpseudocode}
\usepackage{threeparttable}
\usepackage[margin=1in]{geometry}
\usepackage{tikz}
\usetikzlibrary{matrix,calc,arrows.meta,positioning}

\renewcommand{\P}{\mathbf{P}}

\providecommand{\Z}{\mathbf{Z}}

\newcommand{\y}{\mathbf{y}}
\newcommand{\z}{\mathbf{z}}
\let\origc\c
\DeclareRobustCommand\c{\ifmmode\mathbf{c}\else\expandafter\origc\fi}
\let\origd\d
\DeclareRobustCommand\d{\ifmmode\mathbf{d}\else\expandafter\origd\fi}
\let\origu\u
\DeclareRobustCommand\u{\ifmmode\mathbf{u}\else\expandafter\origu\fi}
\let\origd\v
\DeclareRobustCommand\v{\ifmmode\mathbf{v}\else\expandafter\origv\fi}

\providecommand{\ba}{\boldsymbol{\alpha}}

\providecommand{\bg}{\boldsymbol{\gamma}}

\renewcommand{\Pr}{\mathbb{P}}

\providecommand{\Tr}{^{\scriptscriptstyle\top}}

\usepackage[scr=boondox]{mathalpha}

\makeatletter
\@ifpackagelater{mathalpha}{2021/01/01}{%

}{%

}
\makeatother

\providecommand{\one}{\mathbf{1}}
\providecommand{\zero}{\mathbf{0}}

\providecommand{\norm}[1]{\lVert#1\rVert}

\theoremstyle{definition}

\newtheorem{assumption}{Assumption}
\newtheorem{proposition}{Proposition}[section]

\newtheorem{lemma}{Lemma}[section]

\newtheorem{theorem}{Theorem}[section]

\title[TEAM-IV]
{A Novel Approach to Instrumental Variable Estimation: TEAM-IV}

\begin{document}

\author[S. Alberding]{Steven Y. Alberding}

\address{
Department of Biostatistics,
University of Iowa,
Iowa City, Iowa,
USA
}

\email{steven-alberding@uiowa.edu}

 \author[P. Breheny]{Patrick Breheny}
 
 \address{
Department of Biostatistics,
University of Iowa,
Iowa City, Iowa,
USA
}

\email{patrick-breheny@uiowa.edu}
\begin{abstract}
Instrumental-variable (IV) analyses can be undermined when some instruments violate the exclusion restriction through direct effects on the outcome. Existing robust IV methods, including sisVIVE and CIIV, rely on majority- or plurality-validity conditions. We propose TEAM-IV, which targets joint validity by identifying sets of instruments that appear valid together and aggregating them, allowing reliable estimation even when only a small number of candidate instruments are valid (potentially as few as two). TEAM-IV fits MCP-penalized models of the outcome on projected exposure over overlapping instrument subsets. Instruments whose estimated direct effects are frequently shrunk to zero together (suggesting concordant validity status) form candidate “teams” (instrument sets). Teams are ranked by their out-of-sample predictive performance; near-best teams are unioned and used as the instrument set for the IV estimate. Across simulations varying the number, magnitude, and signs of invalid direct effects, instrument correlation, and exposure endogeneity, TEAM-IV generally attains lower absolute estimation error than sisVIVE and CIIV. Relative gains are largest when invalid instruments are prevalent. TEAM-IV is demonstrated in an empirical Mendelian randomization application using the Multi-Ethnic Study of Atherosclerosis (MESA) concerning the effect of LDL cholesterol on carotid intima--media thickness. Two additional MESA-anchored validation studies demonstrate differences between TEAM-IV and sisVIVE under challenging invalid-instrument configurations.
\end{abstract}

\setstretch{1.075}
\maketitle

\pagestyle{headings}



\section{Introduction}\label{intro}
\subsection{Instrumental Variable Estimation and Two-Stage Least-Squares}
\par

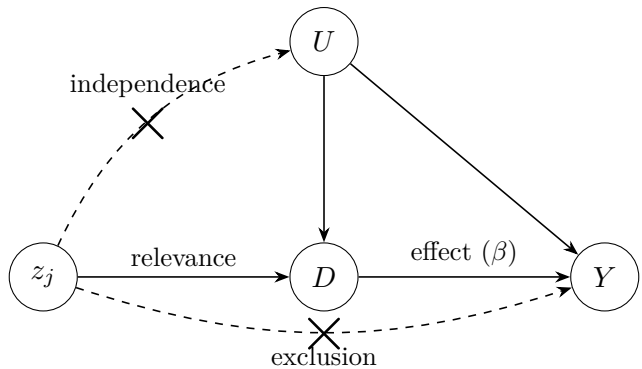
\begin{figure}[t]
\centering

\begin{tikzpicture}[
  >=Stealth,
  node distance=22mm and 28mm,
  var/.style={
    circle,
    draw,
    minimum size=9mm,
    inner sep=1pt,
    align=center
  },
  lab/.style={font=\small},
  req/.style={->, line width=0.6pt},
  bad/.style={->, dashed, line width=0.6pt},
]
\node[var] (Z) {$z_j$};
\node[var, right=of Z] (D) {$D$};
\node[var, right=of D] (Y) {$Y$};
\node[var, above=of D] (U) {$U$};

\draw[req] (Z) -- node[lab, above] {relevance} (D);
\draw[req] (D) -- node[lab, above] {effect ($\beta$)} (Y);
\draw[req] (U) -- (D);
\draw[req] (U) -- (Y);

\draw[bad, bend left=25] (Z) to (U);
\node[lab, xshift=-4.75mm, yshift=4.75mm]
  at ($(Z)!0.5!(U)$) {\Huge $\times$};
\node[lab, xshift=-4.75mm, yshift=9.75mm]
  at ($(Z)!0.5!(U)$) {independence};

\draw[bad, bend right=18] (Z) to (Y);
\node[lab, yshift=-7.5mm]
  at ($(Z)!0.5!(Y)$) {\Huge $\times$};
\node[lab, yshift=-10.5mm]
  at ($(Z)!0.5!(Y)$) {exclusion};

\end{tikzpicture}

\caption{Causal structure for a valid instrumental variable.}
\label{fig:iv-causal-structure}
\end{figure}
Instrumental variables (IV) methods are widely used to estimate causal effects in the presence of unmeasured confounding \citep{bowden2018}. A variable $Z$ is a \emph{valid} instrument for an exposure $D$ and outcome $Y$ if it satisfies three conditions:
\begin{enumerate}
\item \emph{relevance}: $Z$ is associated with $D$;
\item \emph{independence}: $Z$ is independent of unmeasured confounders jointly influencing $D$ and $Y$;
\item \emph{exclusion restriction}: $Z$ affects $Y$ only through $D$, with no direct pathway from $Z$ to $Y$.
\end{enumerate}

When these three conditions are satisfied, it is possible to estimate a causal effect $\beta$ even in the presence of a confounder $U$ that influences both the exposure $D$ and the outcome $Y$. A standard and widely used estimator for the linear instrumental-variables model is two-stage least-squares (2SLS), which first projects the exposure onto the column space of the instrumental variables and then fits an ordinary least-squares regression of the outcome on the resulting fitted exposure \citep{wooldridge2010}. 2SLS works because the instrument-driven variation in the exposure, contained in the instrument space, is independent of the confounder under the IV assumptions, whereas the residual variation in the exposure may be confounded.

\subsection{Invalid Instruments}
The exclusion restriction and independence assumptions are largely untestable \citep{hernan2006}. This is problematic for users of instrumental variable estimation because, if 2SLS is applied using one or more instruments that violate either assumption, the resulting estimator $\widehat\beta_\textrm{2SLS}$ may be biased \citep{kang2016}. In particular, violation of the exclusion restriction can occur when an instrument $Z$ has a direct effect on the outcome $Y$ that does not operate through the exposure $D$. In the Mendelian randomization literature, such direct effects are commonly termed horizontal pleiotropic effects.

\subsection{Mendelian Randomization}

A major application of instrumental variable estimation is Mendelian randomization. This approach uses genetic variants (e.g., single nucleotide polymorphisms -- SNPs) as instrumental variables. Because genetic variants are fixed at conception, they cannot be altered by environmental or behavioral factors occurring after conception, making them natural candidates for instruments satisfying the independence assumption \citep{smith2003, chen2008}.

An illustrative example concerns the impact of alcohol intake on blood pressure. Individuals who inherit the null variant of the ALDH2 gene experience facial flushing and nausea after drinking alcohol and consequently tend to consume less alcohol than individuals with the active genotype. A Mendelian randomization meta-analysis by \citet{chen2008}, synthesizing evidence from studies of Japanese participants, found evidence that greater alcohol consumption causally increases both systolic and diastolic blood pressure. Under the IV assumptions described above, this provides evidence that alcohol consumption has a causal effect on blood pressure.

\subsection{Existing Methods and sisVIVE}

An active area of research over the past decade has concerned IV methods that estimate a causal effect of interest in the presence of instruments that violate the exclusion restriction. Proposed methods in this category include MR-Egger regression and two-stage hard thresholding with voting (TSHT) \citep{bowden2015,guo2018}. MR-Egger was developed primarily for summary-data Mendelian randomization, whereas methods designed for individual-level data include sisVIVE and the confidence interval method for selecting valid instrumental variables (CIIV) \citep{kang2016,windmeijer2021}.

sisVIVE resembles a lasso model with both the exposure and each instrument included as predictors of the outcome, but differs in two crucial respects. First, the squared-loss term is not the squared \(L_2\)-norm of the residual itself, but instead the squared \(L_2\)-norm of the residual after projection onto the column space of the instruments. Second, the penalty is applied only to the \(L_1\)-norm of the direct-effect vector and not to the causal-effect parameter.

A main strength of sisVIVE and CIIV is that, under a two-stage additive linear-effects model and additional regularity conditions, they can consistently estimate the causal effect without requiring all candidate instruments to be valid. Sufficient identification conditions require valid instruments to constitute a majority of the candidate instruments for sisVIVE, whereas CIIV can operate under a plurality condition.

\subsection{TEAM-IV Overview}

TEAM-IV differs from sisVIVE by constructing and evaluating multiple candidate valid instrument sets rather than relying on a single penalized fit. First, the instruments are reordered according to a ridge regression estimate of the direct-effect vector. Next, contiguous windows of the ordered instruments are analyzed using separate minimax concave penalty (MCP) regressions. The resulting MCP estimates are aggregated to group the candidate instruments into putatively valid ``teams.'' Each team is then evaluated according to its cross-validated prediction risk. Finally, the union of the near-best teams is taken as the putative valid instrument set, and the causal effect is estimated using these instruments while treating the remaining instruments as invalid.

\subsection{Contributions}
Our paper adds to the prior literature as follows. We introduce a new team-based selection-and-aggregation procedure that constructs and evaluates multiple putative valid instrument sets for estimating a causal effect in the presence of confounding and an unknown number of invalid instruments.

We show two theoretical results. First, that under a simplified symmetric population regime with homogeneous instrument strengths and compound-symmetric instrument covariance, ordering the instruments according to a ridge regression estimate of the direct-effect vector asymptotically places the valid instruments in a contiguous block. Second, we establish that the window-specific MCP estimator asymptotically assigns a common sign-with-tolerance classification to all valid instruments within a window, even when the corresponding direct-effect estimates are nonzero. 

This paper also details results of a simulation study that compares TEAM-IV against sisVIVE, CIIV, and an oracle 2SLS benchmark. We demonstrate that TEAM-IV generally improves upon sisVIVE in simulation settings where invalid instruments constitute a majority of the candidate instruments, while maintaining comparable performance when a majority are valid. Lastly, we illustrate the practical applicability of TEAM-IV through a real-data Mendelian randomization analysis of the causal effects of LDL cholesterol and HDL cholesterol on carotid far-wall intima-media thickness (IMT), a surrogate marker of cardiovascular disease, using data from the Multi-Ethnic Study of Atherosclerosis (MESA) \citep{bild2002}.

\section{TEAM-IV Method}\label{model}

Consider independent observations \(\{(y_i,d_i,\mathbf z_{i\cdot}^\top)\}_{i=1}^n\), where
\(y_i\in\mathbb R\) is the outcome, \(d_i\in\mathbb R\) is the exposure, and
\(\mathbf z_{i\cdot}\in\mathbb R^L\) is a vector of candidate instruments. We work under the additive linear IV model used by \citealp{kang2016}
\begin{align}
\label{DGP}
y_i = d_i\beta^\star + \mathbf z_{i\cdot}^\top \boldsymbol\alpha^\star + \epsilon_i,
\qquad
\mathbb E[\epsilon_i \mid \mathbf z_{i\cdot}] = 0,
\end{align}
where \(\beta^\star\) is the causal parameter of interest and
\(\boldsymbol\alpha^\star=(\alpha_1^\star,\dots,\alpha_L^\star)^\top\) encodes instrument validity: instrument \(j\) is valid when \(\alpha_j^\star=0\) and invalid when \(\alpha_j^\star\neq 0\). 

We assume the first-stage relationship
\[
d_i = \mathbf z_{i\cdot}^\top \boldsymbol\gamma^\star + \nu_i,
\qquad 
\mathbb E[\mathbf z_{i\cdot}\nu_i]=\mathbf 0,
\]
with \(\gamma_j^\star\neq 0\) for all \(j=1,\dots,L\). \(\operatorname{Cov}(\nu_i,\epsilon_i)\neq 0\) is allowed, and interpreted as encoding endogeneity of $d_i$ with respect to the structural error $\epsilon_i$. Throughout, TEAM-IV requires at least two valid instruments. See Appendix B for details.

\subsection*{Key stages of TEAM-IV}
\begin{enumerate}
\item Order instruments by ridge-based direct-effect estimates
\item Construct teams via within-window minimax concave penalized regression
\item Evaluate candidate teams
\item Estimate the causal effect using instruments in the selected union
\end{enumerate}

\subsection{Ordering instruments by ridge-based direct-effect estimates}
\label{homexposure}
For the first two stages, problems arise when the elements of the true instrument-strength vector $\bg^\star$ are highly heterogeneous. In that case, the projected exposure $\P_\Z \d$ is disproportionately driven by the strongest instruments, so a fitted term $\beta \P_\Z \d$ can absorb direct-effect signal from those strong instruments. As a result, the fitted coefficients on $\Z$ may reflect both exclusion violations and first-stage heterogeneity, leading to inaccurate estimates of both \(\beta\) and \(\ba\).

To reduce heterogeneity in instrument strengths $\bg^\star$ within \(\mathrm{col}(\mathbf Z)\), we replace \(\mathbf d\) by an adjusted exposure vector \(\mathbf d_{\mathrm{adj}}\), constructed so that its projection onto \(\mathrm{col}(\mathbf Z)\) lies in the one-dimensional span of the column sum of \(\mathbf Z\). Specifically,
\[ 
\widehat{\d}_{\mathrm{hom}}:= \mathbf P_{\mathbf Z}\mathbf d_{\mathrm{adj}}
=
\hat\gamma_{\mathrm{sum}}\,\mathbf z_{\mathrm{sum}},
\qquad
\mathbf z_{\mathrm{sum}}=\mathbf Z\mathbf 1_L
\] 
This is equivalent to replacing the projected exposure direct in $\text{col}(\Z)$ with one corresponding to homogenous instrument strengths. The full construction is given in Appendix A.

The purpose of \(\widehat{\d}_{\mathrm{hom}}\) is to suppress heterogeneity in instrument strength, so that the ridge ordering and team construction are mainly driven by direct-effect structure. Without this adjustment, large variation in \(\{|\gamma^\star_j|\}_{j=1}^L \) can distort the fitted coefficients \(\hat{\boldsymbol\alpha}\), causing instruments to separate according to first-stage leverage rather than direct-effect class.

We fit the ridge-penalized projected-exposure model
\[
(\hat{\boldsymbol\alpha},\hat\beta)
=
\arg\min_{\boldsymbol\alpha,\beta}
\frac{1}{2n}
\left\|
\mathbf y-\mathbf Z\boldsymbol\alpha-\widehat{\d}_{\mathrm{hom}}\beta
\right\|_2^2
+
\lambda \|\boldsymbol\alpha\|_2^2.
\]
The ridge penalty yields
a stable, non-sparse estimate $\hat{\boldsymbol\alpha}$ that is suitable for ranking rather than selection.
$\lambda$ is chosen by cross-validation.

We then order instruments by increasing $\hat\alpha_j$, yielding an ordered instrument matrix $\mathbf Z_{\mathrm{ord}}$ with the same column space as $\mathbf Z$.
Under the linear IV model of~\eqref{DGP} and use of $\d_\mathrm{adj}$, instruments that share the same true direct effect $\alpha_j$ will typically have similar fitted coefficients. For example, if all invalid instruments share one value of $\alpha_j = \tau \neq 0$ and all valid instruments have $\alpha_j = 0$, both sets may be shifted away from zero by the same amount, but the valid instruments will still cluster together around one level of $\hat\alpha_j$ and the invalid instruments around another. Thus, the ordering induced by $\{ \hat\alpha_j\}_{j=1}^L$ is intended to place valid instruments in a contiguous block, even when their estimated direct effects are not close to zero.

Appendix B studies population conditions under which the ridge-ordering target is constant within direct-effect classes. In particular, it shows that when the adjusted projected exposure is constrained by construction to the one-dimensional span of \(\Z\one_L\), and the instrument distribution satisfies a symmetric design condition such as compound-symmetric covariance, the population ridge target is constant within each direct-effect class. Under this regime, distinct direct-effect classes induce distinct population ridge levels, so the valid instruments share one common ridge level and each invalid direct-effect class shares its own. In this case, ordering the instruments by \(\{\hat\alpha_j\}_{j=1}^L\) places the valid instruments into a contiguous block.

\subsection{Construction of ``teams'' via within-window minimax concave penalized regression}
\label{mcpteams}
For any ordered set of candidate instruments $1, \hdots, L$, TEAM-IV considers all overlapping contiguous windows of odd cardinality at least 3. Formally, for odd width
$$ w \in \{3, 5, \dots, w_\textrm{final}\}, \qquad w_\textrm{final}= \max \{ w \leq L : w \text{ odd}\},$$ we define window $W(\ell, w) =\{ \ell, \ell+1,\hdots, \ell+w-1\}$ with $\ell = 1,\hdots, L-w+1$.

Thus, for any $L>4$, that the width-\(3\) windows are $\{1,2,3\}, \: \{2,3,4\},$ and so on. The final window will be $\{L-w_\textrm{final}+1,\dots,L\}$. 

Let $W=W(\ell,w)$. Within each window, we consider
\begin{align*}
(\widehat{\boldsymbol\alpha}_W,\widehat\beta_W)
\in
\arg\min_{\boldsymbol\alpha_W,\beta}
\left\{
\frac{1}{2n}
\left\|
\mathbf y
-
\mathbf Z_W\boldsymbol\alpha_W
-
\mathbf P_{\mathbf Z_W}
\widehat{\mathbf d}_{\mathrm{hom}}\beta
\right\|_2^2
+
\sum_{j\in W}
p_{\lambda,\psi}(\alpha_j)
\right\}
\end{align*}
where $p_{\lambda,\psi}$ is the minimax concave penalty (MCP) with tuning parameter $\lambda$ and concavity parameter $\psi$; we fix $\psi = 3$ (the default in package \texttt{ncvreg}; see \citealp{breheny2011}):
\[
p_{\lambda,\psi}(\alpha_j) =
\begin{cases}
\lambda |\alpha_j| - \dfrac{\alpha_j^2}{2\psi}, & \text{if } |\alpha_j| \le \psi \lambda,\\[4pt]
\dfrac{1}{2} \psi \lambda^2, & \text{if } |\alpha_j| > \psi \lambda,
\end{cases}
\]
Local fitting across contiguous windows allows TEAM-IV to compare nearby instruments through penalized direct-effect estimates $\hat\alpha_j$. After all local fits, dyad-wise evidence aggregation across windows stabilizes these comparisons into a final representation of which instruments exhibit similar direct-effect levels. This representation is the final teamness matrix \(T\).

For each window \(W\), let
\[
s_j^{(W)}=\operatorname{sign}\!\bigl(\hat\alpha_j^{(W)}\bigr), \qquad j\in W.
\]
The window-specific contribution \(T^{(W)}\) is defined dyad-wise for \(j\neq k\) by
\[
T_{jk}^{(W)}=
\begin{cases}
1, & s_j^{(W)}=s_k^{(W)}=0,\\
0, & s_j^{(W)}=s_k^{(W)}\in\{-1,+1\},\\
-1, & s_j^{(W)}\neq s_k^{(W)}.
\end{cases}
\]
Diagonal entries are undefined, and dyads not jointly contained in \(W\) receive no contribution from that window. The final teamness matrix is obtained by dyad-wise summation over all windows:
\[
T_{jk}=\sum_{W\in\mathcal W:\,\{j,k\}\subseteq W} T_{jk}^{(W)}.
\]

Each row of final teamness matrix $T$ defines a candidate team $G_j \subseteq \{1,\dots,L \}$. For row $j$ of $T$, we construct 
$G_j = \{j\} \cup \{ k \neq j: T_{jk}>0\}$. That is, instrument \(j\) together with all instruments whose dyad-wise teamness score with \(j\) is positive. Thus, \(G_j\) collects instruments for which the aggregated window evidence favors a common direct-effect level with instrument \(j\). Distinct rows may generate identical candidate teams; duplicate sets are retained only once before cross-validation.

\begin{figure}[t]
\centering
\begin{tikzpicture}[
  >=Stealth,
  cell/.style={draw, minimum width=12mm, minimum height=10mm, align=center},
  lab/.style={font=\small},
  title/.style={font=\small\bfseries}
]

\node[cell] (c11) at (0,0)   {\(0\)};
\node[cell] (c12) at (1.4,0) {\(-1\)};
\node[cell] (c13) at (2.8,0) {\(-1\)};

\node[cell] (c21) at (0,-1.2)   {\(-1\)};
\node[cell] (c22) at (1.4,-1.2) {\(+1\)};
\node[cell] (c23) at (2.8,-1.2) {\(-1\)};

\node[cell] (c31) at (0,-2.4)   {\(-1\)};
\node[cell] (c32) at (1.4,-2.4) {\(-1\)};
\node[cell] (c33) at (2.8,-2.4) {\(0\)};

\draw[line width=0.6pt] (c11.north west) rectangle (c33.south east);

\node[lab] at ($(c11.north)!0.5!(c13.north) + (0,8mm)$)
  {\(s_j=\mathrm{sign}(\widehat\alpha_j^{(W)})\)};
\node[lab] at ($(c11.north)+(0,3mm)$) {\(-1\)};
\node[lab] at ($(c12.north)+(0,3mm)$) {\(0\)};
\node[lab] at ($(c13.north)+(0,3mm)$) {\(+1\)};

\node[lab, rotate=90] at ($(c11.west)!0.5!(c31.west) + (-12mm,0)$)
  {\(s_k=\mathrm{sign}(\widehat\alpha_k^{(W)})\)};
\node[lab] at ($(c11.west)+(-7mm,0)$) {\(-1\)};
\node[lab] at ($(c21.west)+(-7mm,0)$) {\(0\)};
\node[lab] at ($(c31.west)+(-7mm,0)$) {\(+1\)};

\end{tikzpicture}

\caption{Window-specific dyad update rule used to define \(T^{(W)}_{jk}\). Entries indicate the contribution assigned to dyad \((j,k)\) according to the signs of the local MCP-fitted direct effects \(\widehat\alpha_j^{(W)}\) and \(\widehat\alpha_k^{(W)}\).}
\label{fig:teamiv-sign-grid}
\end{figure}

\subsection{Evaluate candidate teams}
\label{teamselect}
For each distinct team $G_j$ returned by the prior step, we evaluate its plausibility as a valid instrument set using a sisVIVE-style cross-validation loss based on regressing $\y$ on $\P_\Z \d$ and the putatively invalid instruments. 

For a specified candidate valid set \(G\subseteq\{1,\ldots,L\}\), let
\(G^c=\{1,\ldots,L\}\setminus G\) denote the corresponding candidate invalid
set. Define
\[
X_G := (1,D_Z,Z_{G^c}).
\]
The associated prediction class is
\[
\mathcal F_G
=
\left\{
x\mapsto c+bD_Z+Z_{G^c}^{\top}\boldsymbol\theta
\right\}.
\]
Equivalently, \(\mathcal F_G\) is the class of linear predictors for \(Y\)
using an intercept, \(D_Z\), and the instruments indexed by \(G^c\) as
direct-effect regressors. The corresponding population prediction risk is
\[
R(G)
:=
\inf_{f\in\mathcal F_G}
\mathbb E\!\left[(Y-f)^2\right].
\]

Let \(G^\star=\{j:\alpha_j^\star=0\}\) denote the oracle valid set. Considering the population risk at $G^\star$, the best predictor in $\mathcal F_{G^\star}$ captures all direct effects of the truly invalid instruments. The resulting residual is $Y-f_{G^\star}^\star=\beta^\star (D-D_Z)+\epsilon$. Notice that there are no omitted direct-effect terms present. For any team $G$ with $G \nsubseteq G^\star$, there will be a term in the residual $Y-f_{G}^\star$ that is attributable to an omitted direct effect from the invalid index (indices) within $G$. This increases the expected squared residual $R(G)$.

To estimate this quantity, we use \(K\)-fold cross-validation. Let \(F_1,\ldots,F_K\) be a partition of \(\{1,\ldots,n\}\) into \(K\) folds. For each fold \(k\), let \(\mathbf y_{F_k}\) denote the observed outcomes in fold \(k\), and let \(\widehat{\mathbf y}^{(-k)}_{F_k}(G)\) denote the corresponding predictions from the model fit on the remaining \(K-1\) folds under candidate team \(G\). We then define
\[
\widehat R_{\mathrm{CV}}(G)
=
\frac{1}{K}\sum_{k=1}^K
\frac{1}{|F_k|}
\left\|
\mathbf y_{F_k}
-
\widehat{\mathbf y}^{(-k)}_{F_k}(G)
\right\|_2^2.
\]

For each distinct candidate team \(G\) returned by the prior step, we evaluate its plausibility as a valid instrument set using a sisVIVE-style cross-validation criterion. Let \(G^c\) denote the corresponding putative invalid set; only \(G^c\) enters the regression as direct-effect covariates.

\begin{algorithmic}[1] \label{alg:cv-sisvive}
\Require Outcome vector \(\mathbf y\in\mathbb R^n\), exposure vector
\(\mathbf d\in\mathbb R^n\), instrument matrix
\(\mathbf Z\in\mathbb R^{n\times L}\), candidate valid-instrument index set
\(G\), corresponding candidate invalid-instrument index set \(G^c\), and
number of folds \(K\)
\State Randomly partition \(\{1,\ldots,n\}\) into \(K\) folds \(F_1,\ldots,F_K\)
\For{\(k \gets 1,\ldots,K\)}
  \State \(\mathcal I_{\mathrm{te}} \gets F_k\),\quad \(\mathcal I_{\mathrm{tr}} \gets \{1,\ldots,n\} \setminus F_k\)
  \State \(\mathbf y_{\mathrm{tr}} \gets \mathbf y_{\mathcal I_{\mathrm{tr}}}\),\quad
         \(\mathbf d_{\mathrm{tr}} \gets \mathbf d_{\mathcal I_{\mathrm{tr}}}\),\quad
         \(\mathbf Z_{\mathrm{tr}} \gets \mathbf Z_{\mathcal I_{\mathrm{tr}},\cdot}\)
  \State \(\mathbf y_{\mathrm{te}} \gets \mathbf y_{\mathcal I_{\mathrm{te}}}\),\quad
         \(\mathbf d_{\mathrm{te}} \gets \mathbf d_{\mathcal I_{\mathrm{te}}}\),\quad
         \(\mathbf Z_{\mathrm{te}} \gets \mathbf Z_{\mathcal I_{\mathrm{te}},\cdot}\)
  \State \(\mathbf P_{\mathbf Z_{\mathrm{tr}}} \gets \mathbf Z_{\mathrm{tr}}
          (\mathbf Z_{\mathrm{tr}}^\top \mathbf Z_{\mathrm{tr}})^{-1} \mathbf Z_{\mathrm{tr}}^\top\)
  \State \(\mathbf P_{\mathbf Z_{\mathrm{te}}} \gets \mathbf Z_{\mathrm{te}}
          (\mathbf Z_{\mathrm{te}}^\top \mathbf Z_{\mathrm{te}})^{-1} \mathbf Z_{\mathrm{te}}^\top\)
  \State \(\mathbf d_{\mathrm{proj,tr}} \gets \mathbf P_{\mathbf Z_{\mathrm{tr}}} \mathbf d_{\mathrm{tr}}\),\quad
         \(\mathbf d_{\mathrm{proj,te}} \gets \mathbf P_{\mathbf Z_{\mathrm{te}}} \mathbf d_{\mathrm{te}}\)
\State \(\mathbf X_{\mathrm{tr}}
\gets
\bigl(
\mathbf 1,\,
\mathbf d_{\mathrm{proj,tr}},\,
\mathbf Z_{\mathrm{tr},\cdot G^c}
\bigr)\)
\State \(\mathbf X_{\mathrm{te}}
\gets
\bigl(
\mathbf 1,\,
\mathbf d_{\mathrm{proj,te}},\,
\mathbf Z_{\mathrm{te},\cdot G^c}
\bigr)\)
  \State Fit least-squares regression
         \(\mathbf y_{\mathrm{tr}} = \mathbf X_{\mathrm{tr}} \boldsymbol\theta + \boldsymbol\varepsilon\)
         to obtain \(\hat{\boldsymbol\theta}\)
  \State \(\hat{\mathbf y}_{\mathrm{te}} \gets \mathbf X_{\mathrm{te}} \hat{\boldsymbol\theta}\)
  \State \(\mathbf r_{\mathrm{te}} \gets \mathbf y_{\mathrm{te}} - \hat{\mathbf y}_{\mathrm{te}}\)
  \State \(\text{fold\_mse}_k \gets \dfrac{1}{|\mathcal I_{\mathrm{te}}|}
         \sum_{i \in \mathcal I_{\mathrm{te}}} (r_{\mathrm{te},i})^2\)
\EndFor
\State \(\text{mean\_mse} \gets \dfrac{1}{K} \sum_{k=1}^K \text{fold\_mse}_k\)
\State \(\text{se\_mean\_mse} \gets \sqrt{\dfrac{\operatorname{Var}(\text{fold\_mse}_1,\ldots,\text{fold\_mse}_K)}{K}}\)
\State \(\text{mean\_rmse} \gets \sqrt{\text{mean\_mse}}\)
\State \(\text{se\_rmse} \gets \dfrac{\text{se\_mean\_mse}}{2\,\text{mean\_rmse}}\) \Comment{Delta method}
\State \Return \(\text{mean\_rmse},\ \text{se\_rmse}\)
\end{algorithmic}

\begin{assumption}[Residual nonredundancy of omitted invalid instruments]
\label{ass:residual-nonredundancy}
For every candidate team \(G\) containing at least one truly invalid instrument, let
\[
A_G=G\cap A^\star
\]
denote the set of invalid instruments omitted from the direct-effect regressors. Let
\[
\widetilde{\mathbf Z}_{A_G}
=
\mathbf Z_{A_G}
-
\Pi_{\mathcal S_G}\mathbf Z_{A_G},
\qquad
\mathcal S_G
=
\operatorname{span}
\left\{
1,D_Z,\mathbf Z_{G^c}
\right\},
\]
where \(\Pi_{\mathcal S_G}\) denotes population linear projection onto \(\mathcal S_G\). Assume that
\[
\boldsymbol\Sigma_{A_G\mid G^c}
:=
\mathbb E
\left[
\widetilde{\mathbf Z}_{A_G}
\widetilde{\mathbf Z}_{A_G}^{\top}
\right]
\]
is positive definite.
\end{assumption}

Assumption~\ref{ass:residual-nonredundancy} is equivalent to requiring that no nonzero linear combination of the omitted invalid instruments lies in the population linear span of \(1\), \(D_Z\), and \(\mathbf Z_{G^c}\). Because
\[
\boldsymbol\alpha_{A_G}^\star\neq\mathbf 0,
\]
it follows that
\[
\left(\boldsymbol\alpha_{A_G}^\star\right)^{\top}
\boldsymbol\Sigma_{A_G\mid G^c}
\boldsymbol\alpha_{A_G}^\star
>0.
\]
Under the structural model and moment conditions stated above, this quantity equals the excess population prediction risk generated by omitting the direct effects associated with \(A_G\). Hence,
\[
R(G)>R(G_0)
\]
for every contaminated candidate \(G\) and every uncontaminated candidate \(G_0\subseteq G^\star\).

For each candidate direct-effect specification, \(K\)-fold cross-validation provides an estimate of its out-of-sample prediction risk \(R(G)\). Suppose that the candidate list contains at least one team \(G_0\subseteq G^\star\) consisting exclusively of valid instruments. Under standard regularity conditions and Assumption~\ref{ass:residual-nonredundancy}, candidate teams containing one or more truly invalid instruments have strictly greater population prediction risk and are therefore therefore expected to receive larger cross-validated prediction loss. By contrast, an uncontaminated candidate \(G\subseteq G^\star\) may attain the same population risk as the oracle set \(G^\star\), because treating additional truly valid instruments as direct-effect covariates does not introduce omitted direct effects. The criterion should therefore be interpreted as ranking candidate teams according to their predictive compatibility with the validity restrictions, rather than as uniquely identifying \(G^\star\). The preceding team-construction step supplies a finite collection of distinct candidate sets, while the cross-validation criterion ranks those candidates according to the predictive consequences of their implied validity restrictions.


\subsection{Estimate the causal effect using instruments in the selected union}

Let
\begin{align*} 
\widehat{\mathcal G}_{\mathrm{near}}=
\left\{G \in \mathcal G :\widehat R_{\mathrm{CV}}(G)
\le
\min_{G' \in \mathcal G}\widehat R_{\mathrm{CV}}(G')
+c\,\operatorname{SE}_{\mathrm{CV}}\right\},
\end{align*}

where \(c=1\) by default and \(\operatorname{SE}_{\mathrm{CV}}\) denotes the estimated fold-wise standard error of the minimum cross-validation risk. We define the final retained instrument set as
\[
\widehat V
=
\bigcup_{G \in \widehat{\mathcal G}_{\mathrm{near}}} G.
\]
This union rule is intended to retain instruments that appear plausibly valid across near-optimal candidate teams. It guards against instability when multiple teams have very similar cross-validation risk. 
In the final estimation step, the instruments in $\widehat V$ are used as excluded instruments, whereas the instruments classified as putatively invalid, $\widehat V^c$, are included as covariates to adjust for their potential direct effects on the outcome. We then estimate the causal effect $\beta$ by two-stage least squares, instrumenting the exposure with $\mathbf Z_{\widehat V}$ while including $\mathbf Z_{\widehat V^c}$ as exogenous adjustment variables in both stages. The resulting two-stage least squares coefficient of the exposure is taken as the TEAM-IV estimation, $\widehat{\beta}_\mathrm{TEAM}$.

\section{Simulation Study}\label{simulation}

We conduct simulations under a range of settings to evaluate the estimation performance of TEAM-IV relative to established methods, including sisVIVE. Estimation performance is measured by the absolute estimation error $| \widehat \beta -\beta^\star|$. We compare TEAM-IV with the oracle two-stage least squares estimator, which uses knowledge of the true valid and invalid instrument sets.

Throughout all settings, the sample size is $n=2000$ and the number of candidate instruments is $L=10$. In the main simulations, the observations $(y_i, d_i, \mathbf{z}_{i\cdot})$, $i=1,\ldots,n$ are generated according to
$\mathbf z_{i\cdot} \sim \mathcal N_L(\mathbf 0, \mathbf I_L)$ and 

\begin{align*}
\left\{
\begin{aligned}
y_i &= \mathbf z_{i\cdot}\Tr \boldsymbol\alpha^{\star} + d_i \beta^{\star} + \epsilon_i \\
d_i &= \mathbf z_{i\cdot}\Tr \boldsymbol\gamma^{\star} + \nu_i
\end{aligned}
\right.\text{,}
\qquad
\begin{pmatrix}
\epsilon_i \\ \nu_i 
\end{pmatrix}
\overset{\text{iid}}{\sim} \mathcal{N}_2 \left(
\begin{bmatrix}
0 \\ 0
\end{bmatrix},
\begin{bmatrix}
1 & \rho_{\epsilon \nu}\\[2pt]
\rho_{\epsilon \nu} & 1
\end{bmatrix}
\right).
\end{align*}

Throughout all simulations, the causal effect is fixed at $\beta^\star =1$. In the main simulations, the first-stage coefficients are generated as
$$\gamma_j^\star \overset{\text{iid}}{\sim} \operatorname{Unif}(0.08,0.14), \quad j=1,\ldots,L.$$ Under the independent-instrument design, this corresponds to an approximate first-stage $F$-statistic of 25 when the exposure is regressed on all candidate instruments.

Define $s=|A^\star|=\norm{\boldsymbol \alpha^\star}_0$. In all simulations, the invalid set $A^\star$ is selected uniformly at random among all subsets of $\{1, \ldots, L\}$ with cardinality $s$. We vary the following settings:
\begin{enumerate}
\item the number of invalid instruments, $s \in \{0,2,4,6,8\}$;
\item the magnitude of the invalid direct effects, for $s>0$, $a_\alpha\in\{0.2,0.8\}$;
\item the endogeneity parameter $\rho_{\epsilon \nu} \in \{0.3, 0.6\}$.
\end{enumerate}
For $j\in A^\star$, we set
$\alpha_j^\star=a_\alpha$,
and for $j\notin A^\star$, we set $\alpha_j^\star=0$. When $s=0$, all candidate instruments are valid, so $\boldsymbol\alpha^\star=\mathbf 0$.

Throughout the simulation study, TEAM-IV refers to the implementation using the 1-SE unioning rule for valid-instrument selection. The final causal effect estimate is obtained by two-stage least squares after treating instruments outside the selected union as covariates. Unless otherwise noted, all reported TEAM-IV results use this implementation.

For each combination of $(s, a_\alpha, \rho_{\epsilon \nu})$, we perform 1000 Monte Carlo replications. The absolute estimation error $| \widehat \beta -\beta^\star|$ for each replication was computed for TEAM-IV, sisVIVE, oracle 2SLS, and the Confidence Interval Method for Selecting Instrument Variables (CIIV) by Windmeijer et al. (2021).

In additional simulations, we generate $\mathbf z_{i\cdot}$ from multivariate normal distributions with correlated instrument covariance structures. Specifically, we consider a compound-symmetric covariance matrix
$$\boldsymbol\Sigma_Z
=
\frac{1}{2}\mathbf I_L
+
\frac{1}{2}\mathbf 1_L\mathbf 1_L^{\top},$$
and an autoregressive covariance matrix satisfying
$$(\boldsymbol\Sigma_Z)_{jk}
=
0.5^{|j-k|},
\qquad
1\leq j,k\leq L.$$

We also consider targeted sensitivity settings that depart from the main simulations in the structure of $\boldsymbol\alpha^\star$ or $\boldsymbol\gamma^\star$. First, we consider mixed-sign invalid direct effects, in which the nonzero entries of $\boldsymbol\alpha^\star$ have equal magnitude but both positive and negative signs. Second, we consider a heterogeneous signed first-stage design, in which the entries of $\boldsymbol\gamma^\star$ vary in magnitude and are not constrained to have a common sign. These sensitivity settings are used to assess whether TEAM-IV remains stable when the invalid direct effects or first-stage coefficients are less sign-coherent than in the main simulation design.

\subsection{Main simulation results under independent instruments}
\begin{figure}[h!]
\begin{center}
\includegraphics[width=\textwidth]{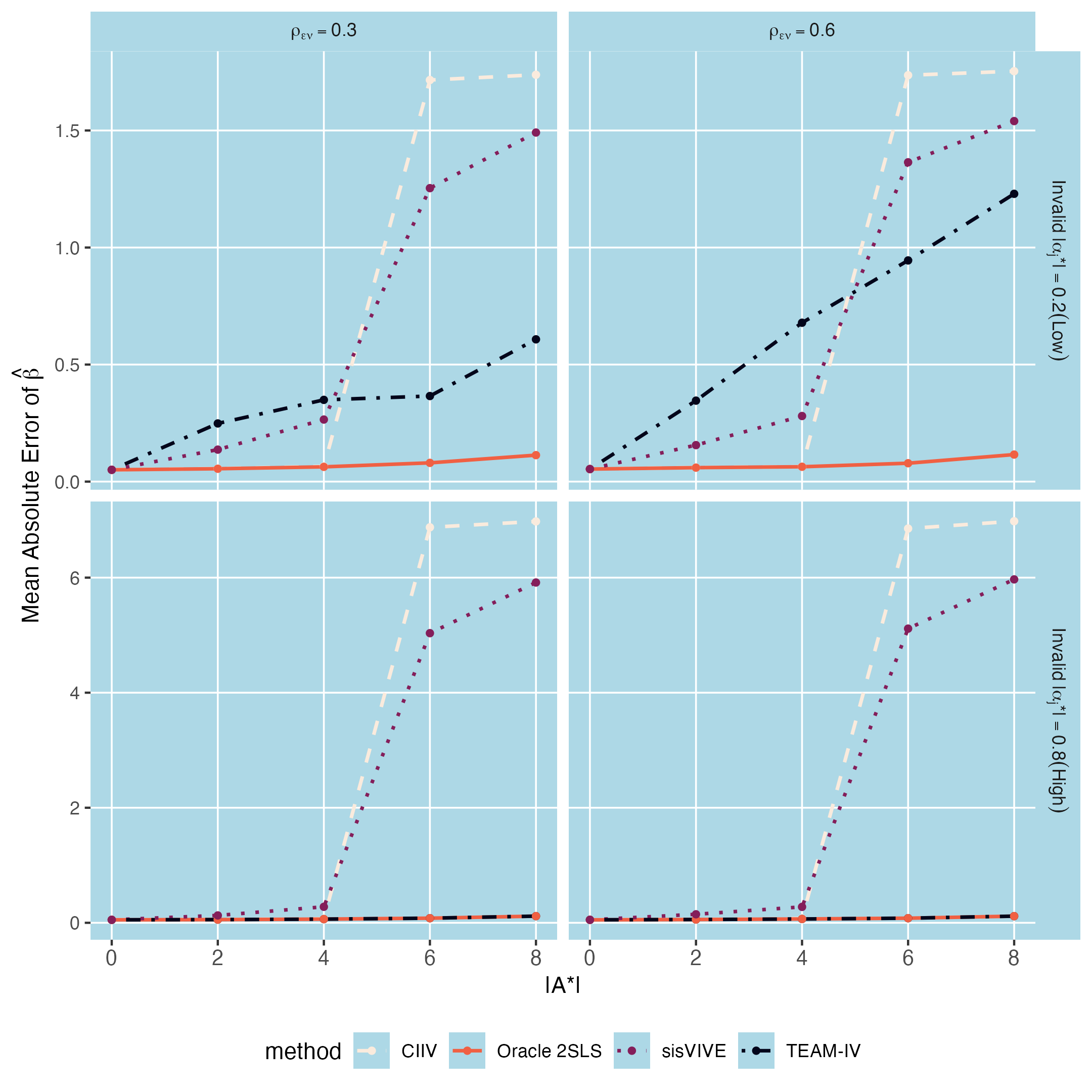}
\caption{Mean absolute error (MAE) of $\widehat\beta$ under the main simulation design with independent candidate instruments. Results are shown by the number of invalid instruments $|A^\star|$, the endogeneity parameter $\rho_{\epsilon\nu}$, and the invalid direct-effect magnitude $a_\alpha$. Each point represents the average across 1000 Monte Carlo replications.}
\label{fig:main-independent-mae}
\end{center}
\end{figure}

Figure~\ref{fig:main-independent-mae} reports the mean absolute estimation error of $\widehat\beta$ under the main simulation design with independent candidate instruments. Results are stratified by the endogeneity parameter $\rho_{\epsilon\nu}$, the magnitude of the invalid direct effects $a_\alpha$, and the number of invalid instruments $|A^\star|$. The compared methods are TEAM-IV, sisVIVE, CIIV, and oracle 2SLS.

Across the settings within the main simulation design, TEAM-IV nearly always produces lower mean absolute estimation error of $\widehat\beta$ than sisVIVE when $|A^\star|\geq 2$. Compared to CIIV, the results are more nuanced, but the instability of CIIV is apparent whenever $|A^\star| \geq 6$, where CIIV produces mean absolute estimation error clearly exceeding both TEAM-IV and sisVIVE. 

Across all four panels, when $|A^\star|=0$, TEAM-IV, sisVIVE, and CIIV all have error approximating the oracle 2SLS level, as expected when all candidate instruments are valid.

When $a_\alpha=0.2$, all methods have relatively small error for $|A^\star| \in \{2,4\}$: TEAM-IV and sisVIVE both show gradually increasing mean absolute error, while the error of CIIV stays extremely close to the oracle 2SLS gold standard. When $|A^\star| \in \{6,8\}$ TEAM-IV performance deteriorates more gradually compared to both sisVIVE and CIIV, although TEAM-IV performance does suffer noticeably relative to oracle 2SLS.

When $a_\alpha=0.8$, the separation between methods is larger. CIIV and sisVIVE exhibit sharp increases in estimation error once the number of invalid instruments satisfies $|A^\star| \in \{6,8\}$. In contrast, TEAM-IV seems to attain the mean estimation error of the oracle 2SLS estimator across the range of $|A^\star|$, indicating robustness to strong invalid direct effects.

Increasing $\rho_{\epsilon \nu}$ from $0.3$ to $0.6$ has a comparatively smaller effect on the qualitative ranking of the methods than increasing $a_\alpha$ or $|A^\star|$. That is, oracle 2SLS is always the best method, TEAM-IV is next best, and sisVIVE and CIIV tend to have more degraded performance when the majority of instruments are invalid.

These results support the advantage of TEAM-IV relative to sisVIVE/CIIV: it remains more stable when the candidate instrument set becomes increasingly contaminated by invalid instruments. The results also show that the relative performance differences are driven more by the number and magnitude of the invalid direct effects than by the endogeneity level.

The main design uses independent candidate instruments, common-sign invalid direct effects, and exclusively positive first-stage coefficients. The following subsections therefore examine TEAM-IV's behavior under correlated instruments, mixed-sign invalid direct effects, and heterogeneous/mixed-sign first-stage coefficients.

\subsection{Robustness to correlated instruments}
\begin{figure}[h!]
\begin{center}
\includegraphics[width=\textwidth]{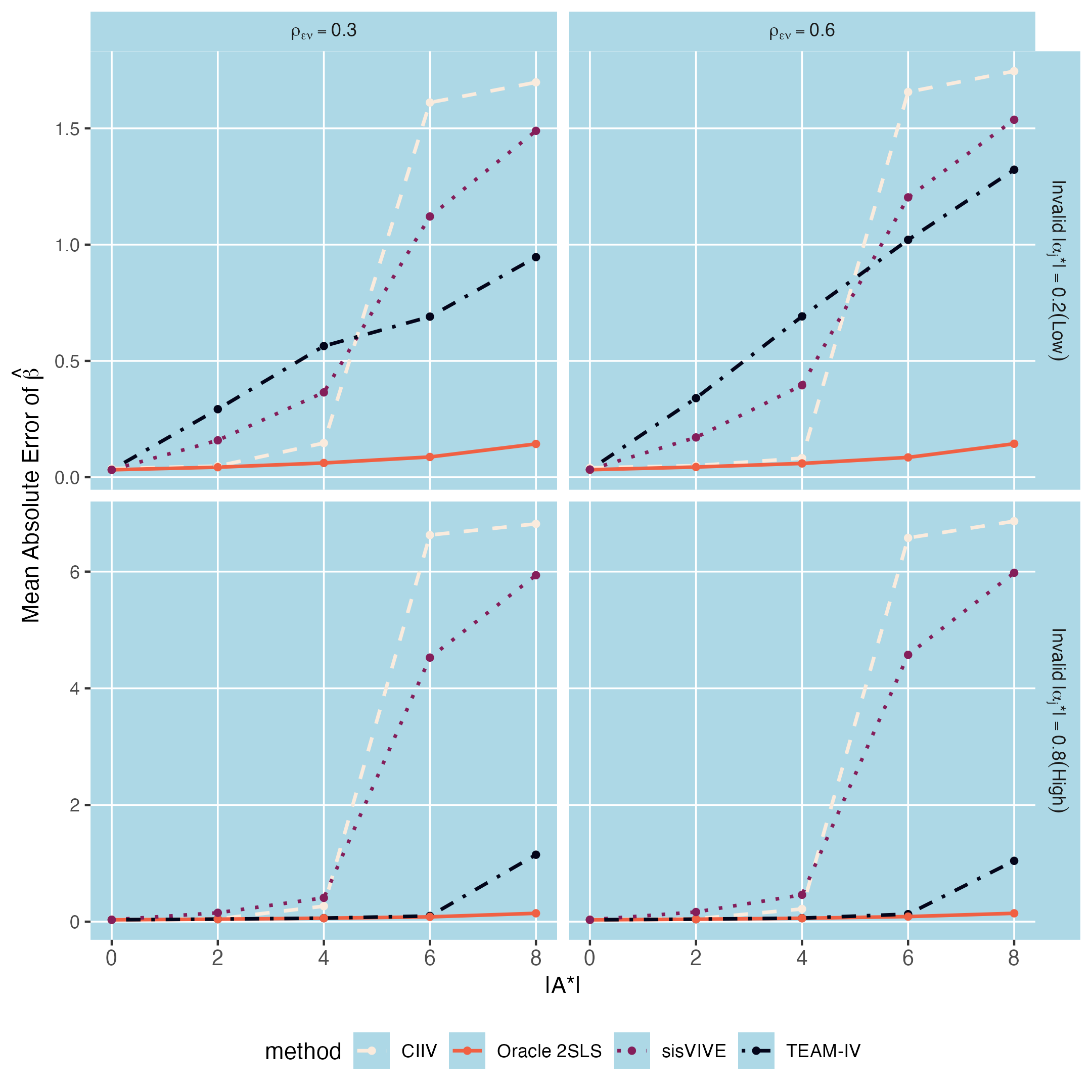}
\caption{Mean absolute error (MAE) of $\widehat\beta$ under the simulation design with autoregressive instrument covariance and positive invalid direct effects. Results are shown by the number of invalid instruments \(|A^\star|\), the endogeneity parameter \(\rho_{\epsilon\nu}\), and the invalid direct-effect magnitude \(a_\alpha\). Each point represents the average across 1000 Monte Carlo replications.}
\label{fig:ar1-mae}
\end{center}
\end{figure}

\begin{figure}[h!]
\begin{center}
\includegraphics[width=\textwidth]{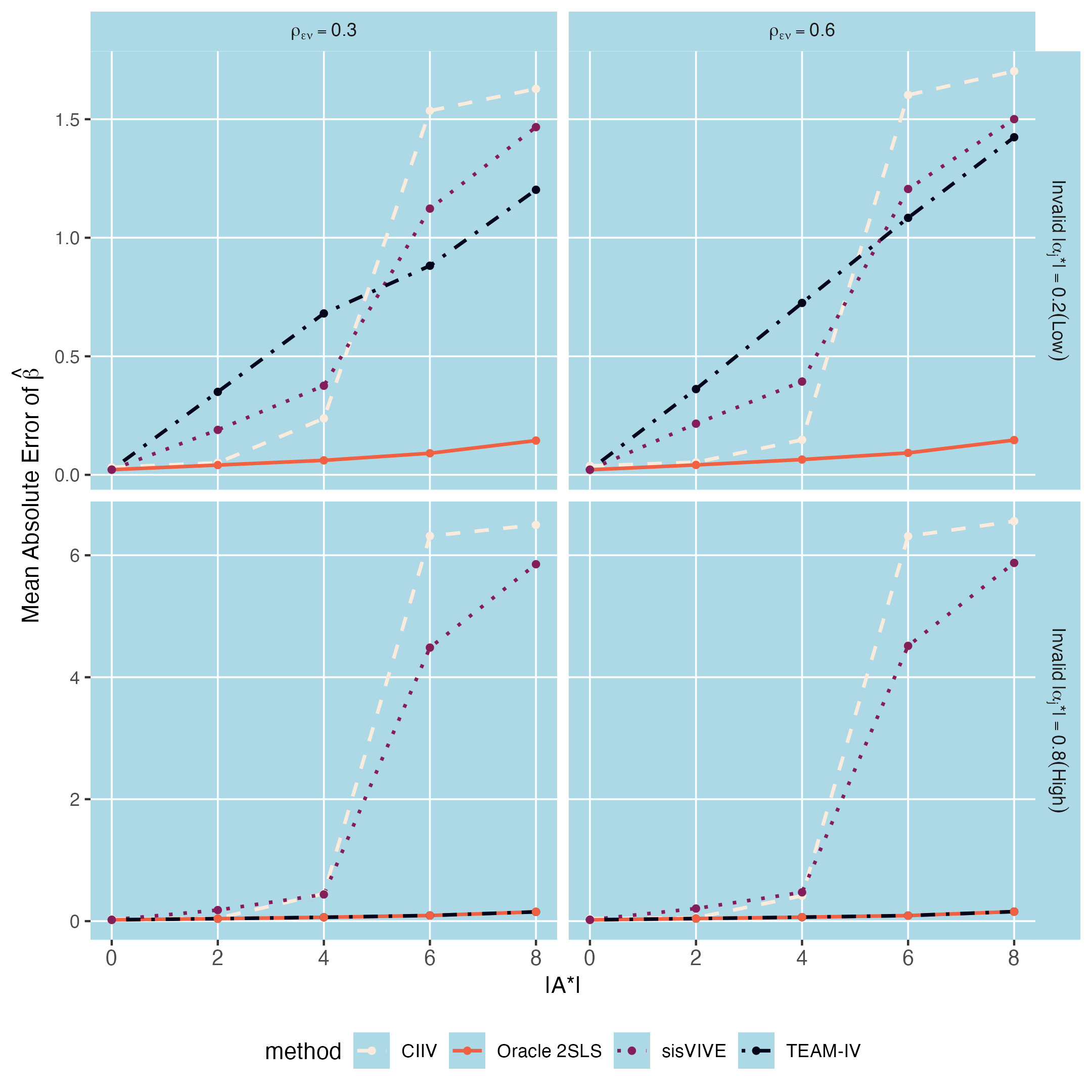}
\caption{Mean absolute error (MAE) of $\widehat\beta$ under the simulation design with compound-symmetric instrument covariance and positive invalid direct effects. Results are shown by the number of invalid instruments \(|A^\star|\), the endogeneity parameter \(\rho_{\epsilon\nu}\), and the invalid direct-effect magnitude \(a_\alpha\). Each point represents the average across 1000 Monte Carlo replications.}
\label{fig:cs-mae}
\end{center}
\end{figure}

Figures~\ref{fig:ar1-mae} and~\ref{fig:cs-mae} report the mean absolute estimation error of $\widehat\beta$ when the candidate instruments are generated from correlated multivariate normal distributions. Figure~\ref{fig:ar1-mae} uses an autoregressive covariance structure with $(\Sigma_Z)_{jk}=0.5^{|j-k|}$, while Figure~\ref{fig:cs-mae} uses a compound-symmetric covariance structure with off-diagonal correlation $0.5$. All other features of the main simulation design are held fixed.

The qualitative ranking of the methods is similar to that observed under independent instruments. Oracle 2SLS has the smallest mean absolute error across settings, TEAM-IV generally remains the best-performing non-oracle method, and sisVIVE and CIIV exhibit substantially larger errors when invalid instruments constitute a majority of candidate instruments.

Under the autoregressive covariance structure, TEAM-IV generally performs comparably to sisVIVE when the majority rule holds (i.e., $|A^\star| \in \{0,2,4\}$). When the majority rule is violated and $|A^\star| \in \{6, 8\}$, sisVIVE and CIIV exhibit very large mean absolute errors while TEAM-IV suffers a more limited deterioration in performance. This pattern indicates that TEAM-IV retains robustness when instrument correlations decay with distance.

Under the compound-symmetric covariance structure, the same broad pattern persists. The correlated design makes the estimation problem more difficult in some low-direct-effect, majority-invalid settings, where TEAM-IV is farther from the oracle 2SLS benchmark. Nonetheless, when $a_\alpha=0.8$, TEAM-IV continues to avoid the large error increases observed for sisVIVE and CIIV when $|A^\star| \in \{6,8\}$.

The autoregressive and compound-symmetric designs lead to similar qualitative conclusions, despite representing different forms of instrument correlation. TEAM-IV performance is therefore not driven solely by the independent-instrument assumption in the main simulation design.

\subsection{Sensitivity to sign incoherence}
\begin{figure}[h!]
\begin{center}
\includegraphics[width=\textwidth]{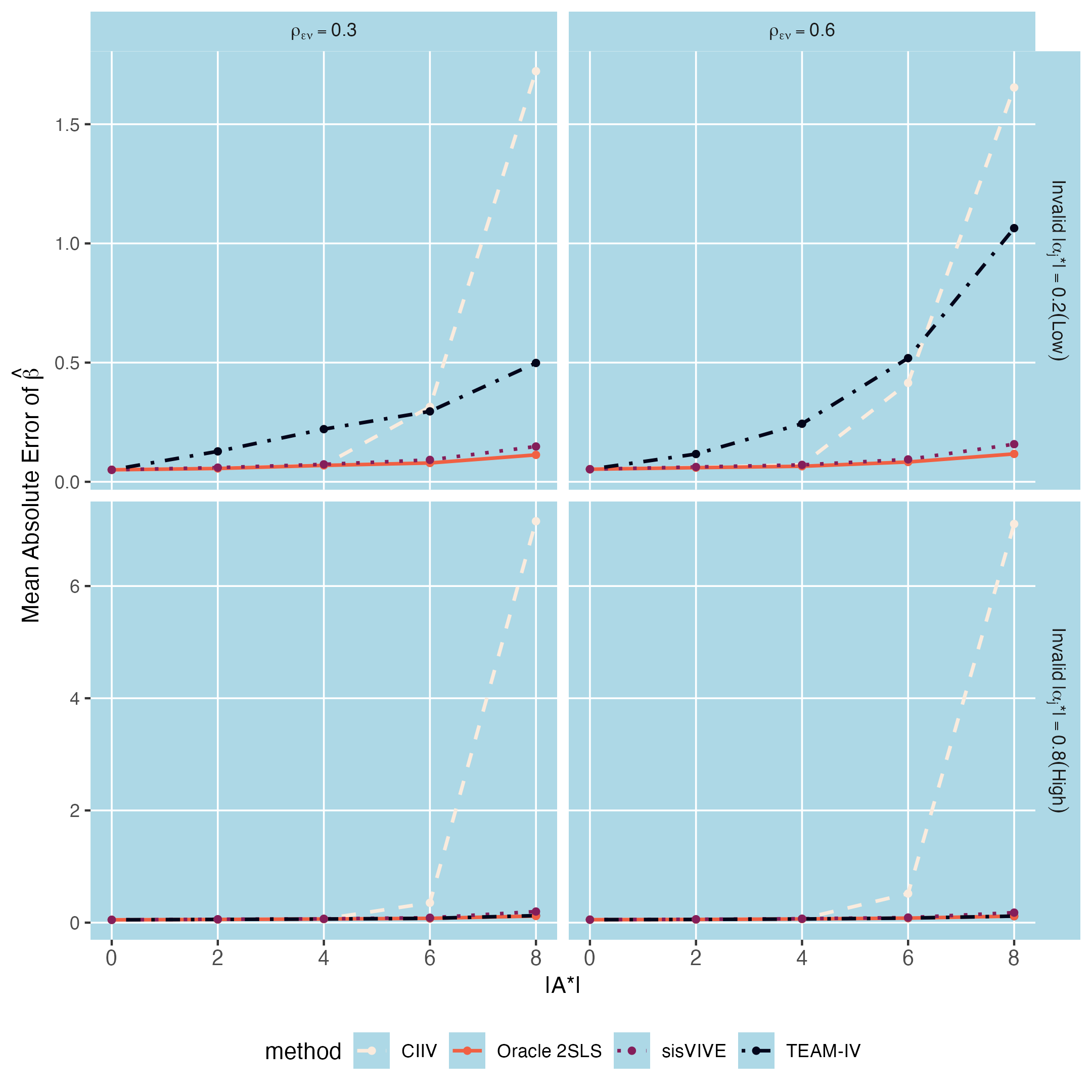}
\caption{Mean absolute error (MAE) of $\widehat\beta$ under the simulation design with independent instrument covariance and invalid direct effects split evenly between positive and negative signs. Results are shown by the number of invalid instruments \(|A^\star|\), the endogeneity parameter \(\rho_{\epsilon\nu}\), and the invalid direct-effect magnitude \(a_\alpha\). Each point represents the average across 1000 Monte Carlo replications.}
\label{fig:alpmix-independent-mae}
\end{center}
\end{figure}

Figure~\ref{fig:alpmix-independent-mae} reports the mean absolute estimation error of $\widehat\beta$ when the candidate instruments are independent and invalid direct effects are split evenly between positive and negative signs. This design relaxes the common-sign direct-effect structure used in the main simulations while retaining the independent instrument covariance structure and the positive first-stage coefficients $\gamma_j^\star$, $j=1,\ldots,L$.

The aim of this sensitivity analysis is to determine whether TEAM-IV performance depends on the sign coherence of the invalid direct effects. In the main simulation design, all invalid instruments have direct effects with the same sign. Here, the nonzero entries of $\boldsymbol \alpha^\star$ have equal magnitude but are evenly divided between positive and negative signs.

The principal effect of using mixed-sign invalid direct effects is that sisVIVE performs substantially better, consistently approximating oracle 2SLS under independent instruments, than in the corresponding common-sign setting. CIIV also benefits from the splitting of invalid direct effects, as its efficacy is contingent on a plurality rule, that can still hold in the case of $|A^\star|=6$ and $|G^\star|=4$, since the direct effects of the invalid instruments fall into two distinct nonzero clusters. TEAM-IV, by contrast, is not materially degraded by the loss of sign coherence and continues to have low mean absolute estimation error across most settings. 

When $a_\alpha=0.2$, TEAM-IV performs slightly to moderately worse than sisVIVE and CIIV when $|A^\star| \in \{2, 4\}$. When $|A^\star|=8$, CIIV error increases dramatically due to violation of the plurality rule when only two instruments are valid, while TEAM-IV suffers a more gradual falloff in performance.

When $a_\alpha=0.8$, the mixed-sign direct-effect design sharply distinguishes CIIV from the other methods. CIIV exhibits large increases in mean absolute estimation error when $|A^\star|=8$, whereas TEAM-IV and sisVIVE remain close to the oracle 2SLS benchmark across all settings of $|A^\star|$ and $\rho_{\epsilon \nu}$. 

Overall, the mixed-sign direct-effect sensitivity analysis indicates that TEAM-IV does not require direct effects to have a common sign in order to remain stable. The larger change is observed for sisVIVE, whose performance improves when there is a balance of positive and negative direct effects. Thus, sign incoherence in $\boldsymbol\alpha^\star$ reduces separation between TEAM-IV and sisVIVE, but it does not create a failure mode for TEAM-IV

The remaining sensitivity analyses examine whether the same qualitative behavior persists when first-stage coefficients $\gamma_j^\star$ are allowed to vary in both sign and magnitude. See Supplementary Materials.

\section{MESA application and data-anchored validation}
\label{data-analysis}

\begin{table}[!htbp]
\centering
\caption{Structure and objectives of the empirical MESA application and
data-anchored validation studies.}
\label{tab:mesa_analysis_structure}

\small
\setlength{\tabcolsep}{4pt}
\renewcommand{\arraystretch}{1.20}

\begin{tabularx}{\textwidth}{
    @{}
    >{\raggedright\arraybackslash}p{0.19\textwidth}
    >{\centering\arraybackslash}p{0.05\textwidth}
    >{\centering\arraybackslash}p{0.05\textwidth}
    >{\centering\arraybackslash}p{0.13\textwidth}
    >{\centering\arraybackslash}p{0.13\textwidth}
    >{\raggedright\arraybackslash}X
    @{}
}
\toprule
\textbf{Component}
& \(\boldsymbol{y}\)
& \(\boldsymbol{d}\)
& \(\boldsymbol{Z}\)
& \textbf{Ground truth}
& \textbf{Main scientific purpose} \\
\midrule

Empirical MESA analysis
& Obs
& Obs
& Obs
& No
& Estimate the effect of LDL cholesterol on carotid IMT and identify
candidate invalid instruments. \\

Outcome-informed pseudo-instrument stress test
& Obs
& Obs
& half obs, half pseudo
& Partial only
& Evaluate detection of deliberately outcome-associated
pseudo-instruments while retaining the observed outcome and exposure. \\

MESA-anchored semi-synthetic outcome analysis
& Sim
& Obs
& Obs
& Yes
& Evaluate causal-effect estimation and recovery of the known
invalid-instrument set under a prespecified outcome-generating model. \\

\bottomrule
\end{tabularx}
\begin{minipage}{\textwidth}
\footnotesize
\textit{Notes:} Obs, observed; Sim, simulated. Partial ground truth means
that the pseudo-instruments are known to have been constructed using the
outcome, whereas the validity of the observed MESA SNPs and the true causal
effect remain unknown. 
\end{minipage}
\end{table}

\begin{table}
\caption{Sample characteristics}
\label{tab:sample_characteristics}
\begin{tabular}{lc}
\toprule
\textbf{Characteristic} & \textbf{Overall; N = 5,602}\\
\midrule
Age, years& 62.2 (10.3)\\
Gender& \\
\hspace{1em}Female & 2,907 (51.9\%)\\
\hspace{1em}Male & 2,695 (48.1\%)\\
Race/ethnicity& \\
\addlinespace
\hspace{1em}White & 2,265 (40.4\%)\\
\hspace{1em}Black & 1,375 (24.5\%)\\
\hspace{1em}Hispanic & 1,272 (22.7\%)\\
\hspace{1em}Chinese American & 690 (12.3\%)\\
LDL cholesterol, mg/dL& 117.2 (31.5)\\
\addlinespace
Mean common carotid far-wall IMT, micrometers& 740.4 (208.2)\\
Lipid-lowering medication use& \\
\hspace{1em}No & 4,696 (83.9\%)\\
\hspace{1em}Yes & 904 (16.1\%)\\
Current cigarette smoking& \\
\addlinespace
\hspace{1em}No & 4,891 (87.3\%)\\
\hspace{1em}Yes & 711 (12.7\%)\\
Diabetes status& \\
\hspace{1em}Normal & 4,119 (73.6\%)\\
\hspace{1em}Impaired fasting glucose (IFG) & 784 (14.0\%)\\
\addlinespace
\hspace{1em}Treated diabetes & 554 (9.9\%)\\
\hspace{1em}Untreated diabetes & 141 (2.5\%)\\
Body mass index, kg/m²& 28.3 (5.5)\\
Systolic blood pressure, mmHg& 126.3 (21.5)\\
\bottomrule
\multicolumn{2}{l}{\rule{0pt}{1em}\textsuperscript{1} Mean (SD); n (\%)}\\
\end{tabular}
\end{table}
\begin{table}[!h]
\centering
\caption{\label{tab:tab:semisynthetic-comparison}Performance of TEAM-IV and sisVIVE in the MESA-anchored semi-synthetic outcome experiment.}
\centering
\begin{threeparttable}
\begin{tabular}[t]{lrrrrrr}
\toprule
Method & nrep & Bias & Mean absolute error & Sensitivity & Specificity & PPV\\
\midrule
TEAM-IV & 50 & -0.217 & 2.323 & 1 & 1.000 & 1.000\\
sisVIVE & 50 & 1.816 & 2.806 & 1 & 0.873 & 0.929\\
\bottomrule
\end{tabular}
\begin{tablenotes}
\item Sensitivity is the proportion of truly invalid instruments classified as invalid. Specificity is the proportion of truly valid instruments classified as valid. PPV is the proportion of instruments classified as invalid that were truly invalid.
\end{tablenotes}
\end{threeparttable}
\end{table}

\begin{table}[!htbp]
\centering
\caption{Performance under the outcome-informed pseudo-instrument stress
test across 50 replications.}
\label{tab:pseudo_instrument_results}

\small
\setlength{\tabcolsep}{4pt}
\renewcommand{\arraystretch}{1.20}

\begin{tabularx}{\textwidth}{
    @{}
    >{\raggedright\arraybackslash}p{0.19\textwidth}
    >{\centering\arraybackslash}p{0.13\textwidth}
    >{\centering\arraybackslash}p{0.13\textwidth}
    >{\centering\arraybackslash}p{0.18\textwidth}
    >{\centering\arraybackslash}X
    @{}
}
\toprule
\textbf{Method}
& \textbf{Mean \(\widehat{\beta}\)}
& \textbf{SD of \(\widehat{\beta}\)}
& \textbf{Pseudo-SNP detection}
& \textbf{Observed SNPs retained} \\
\midrule

TEAM-IV
& \(1.21\)
& \(0.46\)
& \(100.0\%\)
& \(100.0\%\) \\

sisVIVE
& \(7.77\)
& \(0.34\)
& \(1.2\%\)
& \(99.5\%\) \\

\bottomrule
\end{tabularx}

\vspace{0.5em}

\begin{minipage}{\textwidth}
\footnotesize
\textit{Note:}
The observed outcome and exposure were held fixed across replications;
variation therefore reflects repeated construction of outcome-informed
pseudo-instruments. Pseudo-SNP detection is the proportion of constructed
pseudo-instruments classified as invalid, whereas observed-SNP retention is
the proportion of observed MESA SNPs classified as valid. The individual first-stage \(F\)-statistics of the accepted pseudo-SNPs
ranged from 42.33 to values exceeding 100. The observed eight-SNP instrument set had a
joint first-stage \(F\)-statistic of 2.16. Because the
pseudo-instruments were constructed using the observed outcome, this
analysis constitutes an adversarial stress test rather than a simulation
with known causal ground truth.
\end{minipage}

\end{table}

\begin{table}[!h]
\centering
\caption{Summary of instruments used in the analysis of LDL cholesterol and carotid intima-media thickness.}
\label{tab:snp_summary}
\centering
\resizebox{\ifdim\width>\linewidth\linewidth\else\width\fi}{!}{
\begin{threeparttable}
\begin{tabular}[t]{lcccc}
\toprule
Instrument & Major alleles & Heterozygote & Minor alleles & MAF (SE)\\
\midrule
rs11206510 & 4205 (75.1\%; A/A) & 1266 (22.6\%; A/G) & 130 (2.3\%; G/G) & 0.136 (0.003)\\
rs3846663 & 2310 (41.2\%; C/C) & 2535 (45.3\%; C/T) & 757 (13.5\%; T/T) & 0.361 (0.005)\\
rs1501908 & 2134 (38.1\%; C/C) & 2487 (44.4\%; C/G) & 981 (17.5\%; G/G) & 0.397 (0.005)\\
rs4420638 & 4017 (71.8\%; T/T) & 1442 (25.8\%; T/C) & 136 (2.4\%; C/C) & 0.153 (0.003)\\
rs10401969 & 4501 (80.3\%; A/A) & 1022 (18.2\%; A/G) & 79 (1.4\%; G/G) & 0.105 (0.003)\\
\addlinespace
rs6102059 & 2439 (43.6\%; C/C) & 2466 (44.0\%; C/T) & 695 (12.4\%; T/T) & 0.344 (0.004)\\
rs964184 & 3604 (64.4\%; G/G) & 1748 (31.2\%; G/C) & 247 (4.4\%; C/C) & 0.200 (0.004)\\
rs2650000 & 2749 (49.1\%; G/G) & 2243 (40.0\%; G/T) & 610 (10.9\%; T/T) & 0.309 (0.004)\\
\bottomrule
\end{tabular}
\begin{tablenotes}
\item \textit{Note: } 
\item Genotype entries are reported as count (percentage among participants with a nonmissing genotype; genotype). MAF denotes minor-allele frequency. The MAF standard error was calculated as sqrt[MAF(1-MAF)/(2n)], where n is the number of participants with a nonmissing genotype for that instrument.
\end{tablenotes}
\end{threeparttable}}
\end{table}


\begin{figure}[!htbp]
    \centering
    \includegraphics[
        width=0.82\textwidth
    ]{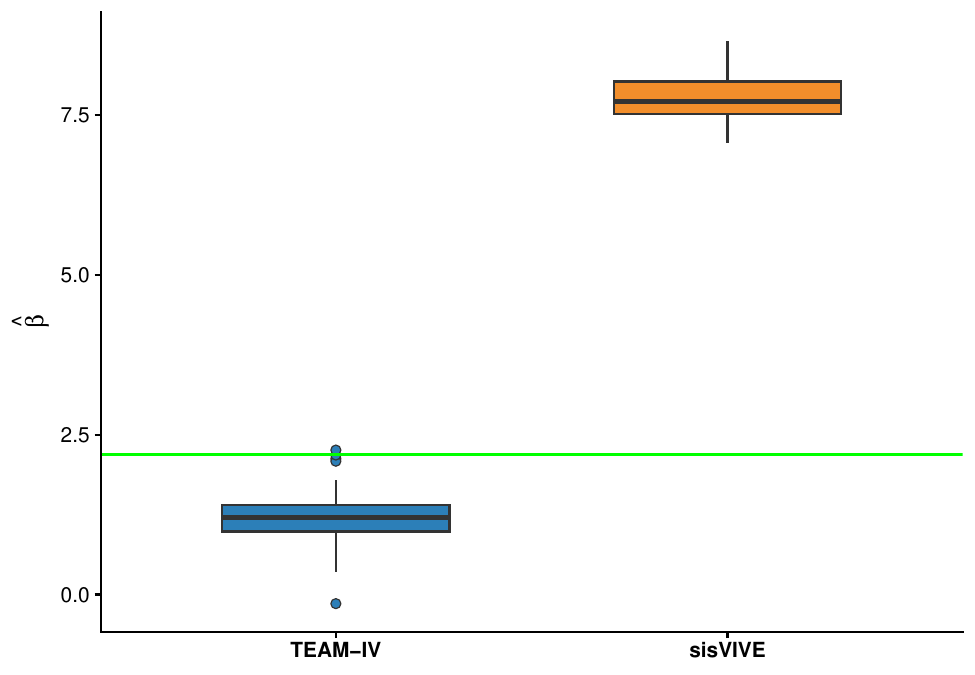}

    \caption{
    Distribution of estimated LDL-C effects across 50 repetitions of
    the outcome-informed pseudo-instrument stress test. In each
    repetition, the instrument set contained eight observed MESA SNPs
    and eight newly constructed pseudo-SNPs, so that 50\% of the
    instruments were pseudo-instruments. The observed carotid IMT
    outcome, LDL-C exposure, covariates, and observed SNPs were held
    fixed across repetitions. Pseudo-SNPs were generated with target
    minor-allele frequency \(0.30\), \(\rho_D=0.30\), and
    \(\rho_Y=0.20\), and were required to have a conditional
    first-stage \(F\)-statistic of at least 10, a direct-association
    \(t\)-statistic of at least 3, and a maximum empirical pairwise
    \(r^2\) of 0.10 with the other instruments. The horizontal green
    line denotes the empirical reference estimate obtained using only
    the eight observed MESA SNPs,
    \(\widehat{\beta}_{\mathrm{MESA}}=2.19\).
    This reference estimate is not known causal ground truth. 
    The coefficients are expressed as micrometers of carotid IMT per
\(1~\mathrm{mg/dL}\) increase in LDL-C.
    }
    \label{fig:pseudo_instrument_beta_hat}
\end{figure}

We analyzed individual-level data from the Multi-Ethnic Study of
Atherosclerosis (MESA), a population-based prospective cohort designed to
investigate subclinical cardiovascular disease in a racially and ethnically
diverse population \citep{bild2002}. The complete-case analytic sample
included \(n=5{,}602\) participants with measurements of LDL cholesterol,
mean common carotid far-wall intima--media thickness (IMT), relevant
covariates, and candidate genetic instruments. We used these data both to
illustrate the application of TEAM-IV in an empirical setting and to conduct
two complementary, data-anchored validation studies.

Table~\ref{tab:mesa_analysis_structure} summarizes the three components of
the analysis. First, we applied TEAM-IV and sisVIVE to the observed MESA
outcome, exposure, and genotype data. Second, we conducted an adversarial
stress test in which half of the instrument matrix consisted of observed
MESA SNPs and half consisted of outcome-informed, genotype-like
pseudo-instruments constructed to be relevant for LDL cholesterol and
strongly associated with carotid IMT conditional on the exposure and
covariates. Because these pseudo-instruments were constructed using the
observed outcome, this analysis evaluates sensitivity to deliberately
outcome-associated instruments but does not provide complete structural
ground truth. Third, we retained the observed LDL cholesterol measurements,
genotypes, and covariates but generated a semi-synthetic outcome under a
prespecified causal effect and direct-effect vector. This final analysis
provides known ground truth for evaluating both causal-effect estimation and
invalid-instrument identification.

\subsection{MESA sample, phenotypes, genotypes, and covariates}

We used individual-level data from the Multi-Ethnic Study of
Atherosclerosis (MESA), a prospective cohort study designed to investigate
the development and progression of subclinical cardiovascular disease in a
racially and ethnically diverse population \citep{bild2002}. The outcome,
\(\mathbf{y}\), was mean common carotid far-wall intima--media thickness (IMT),
measured by ultrasound. Carotid IMT is a noninvasive marker of subclinical
vascular disease, and greater common carotid IMT has been associated with
higher subsequent risk of cardiovascular events
\citep{lorenz2007,polak2011}. The original MESA measurement was recorded in
millimeters and was multiplied by \(1{,}000\) for analysis in micrometers.

The exposure, \(\mathbf{d}\), was LDL cholesterol (LDL-C), measured in
\(\mathrm{mg/dL}\) at MESA Exam~1. LDL-C was selected because extensive
epidemiologic, genetic, and clinical-trial evidence supports a causal and
cumulative role of LDL particles in atherosclerotic cardiovascular disease
\citep{ference2017}. Candidate instruments were obtained from the
LDL-C genome-wide association studies of \citet{kathiresan2009} and
\citet{teslovich2010}. We matched the reported index variants to the
available MESA genotype data and retained eight SNPs that had been associated
with LDL-C at the conventional genome-wide significance threshold
\(P<5\times10^{-8}\). The retained instruments and their characteristics are
reported in Table~\ref{tab:snp_summary}. Genotypes were coded as effect-allele
counts taking values \(0\), \(1\), or \(2\).

The primary analysis adjusted for age, sex, self-reported race, site,
and the first five genotype-derived ancestry principal components.
We restricted the analysis to participants with complete information on
carotid IMT, LDL-C, all eight instruments, and the adjustment covariates,
yielding a final analytic sample of \(n=5{,}602\). Characteristics of this
complete-case sample are presented in
Table~\ref{tab:sample_characteristics}.

\subsection{Empirical analysis of LDL cholesterol and carotid IMT}
The empirical analysis included \(5{,}602\) MESA participants with
complete data on the exposure, outcome, instruments, and adjustment
covariates. The exposure was Exam 1 LDL cholesterol, measured in
\(\mathrm{mg/dL}\). The outcome was mean common carotid far-wall
intima--media thickness, converted from millimeters
to micrometers by multiplying the recorded value by \(1{,}000\).
The eight genome-wide-significant LDL-C-associated SNPs listed in
Table~\ref{tab:snp_summary} were used as candidate instruments.

TEAM-IV and sisVIVE were each adjusted for age, sex, self-reported
race, MESA study site, and the first five genotype-derived
ancestry principal components. These covariates were included in both
the exposure and outcome equations. Participant characteristics for
the complete-case sample are reported in
Table~\ref{tab:sample_characteristics}.

Both methods retained all eight candidate instruments and estimated a
zero vector of direct SNP effects,
\[
\widehat{\boldsymbol{\alpha}}=\boldsymbol{0}_{8}.
\]
TEAM-IV and sisVIVE therefore produced the same point estimate, after
rounding to two decimal places,
\[
\widehat{\beta}=2.19
\quad
\text{\(\mu\mathrm{m}\) per \(\mathrm{mg/dL}\).}
\]
Under the instrumental-variable assumptions, the estimate corresponds to a \(2.19~\mu\mathrm{m}\) increase in mean common carotid far-wall IMT per \(1~\mathrm{mg/dL}\) increase in LDL-C. Neither implementation provides an analytic standard error or
confidence interval for the post-selection effect estimate; therefore,
we report these results as point estimates only.
For reference, a one-standard-deviation increase in LDL-C
(\(31.5~\mathrm{mg/dL}\)) corresponds to an estimated
\(69.0~\mu\mathrm{m}\) increase in carotid IMT.

The zero direct-effect estimates indicate that neither method selected
any of the eight SNPs as invalid. This finding does not establish that
all eight instruments satisfy the exclusion restriction, because the
true direct effects are unknown in the empirical analysis.

The eight SNPs were jointly associated with LDL-C
(\(F=2.16\), \(P=0.027\)), but explained only \(0.31\%\) of the
residual exposure variation after covariate adjustment
(partial \(R^2=0.0031\)). The joint first-stage \(F\)-statistic was
well below the conventional instrument-strength benchmark of 10, indicating
that weak-instrument bias is an important limitation of the empirical
analysis. Accordingly, the TEAM-IV and sisVIVE estimates are
interpreted as exploratory estimates.

\subsection{Adversarial validation using outcome-informed
pseudo-instruments}

We conducted an adversarial stress test to evaluate whether TEAM-IV and
sisVIVE could identify deliberately constructed, outcome-associated
pseudo-instruments. The observed carotid IMT outcome \(\mathbf{y}\), LDL-C
exposure \(\mathbf{d}\), participant covariates \(\mathbf{C}\), and eight MESA SNPs were
held fixed throughout this analysis. In each replication, we generated
eight new genotype-like pseudo-SNPs and appended them to the eight
observed SNPs. Each augmented instrument matrix therefore contained 16
candidate instruments, of which 50\% were observed MESA SNPs and 50\%
were outcome-informed pseudo-instruments.

\paragraph{Construction of the pseudo-instruments.}

Let
\[
\mathbf{Q}=[\boldsymbol{1},\mathbf{C}],
\]
where \(\mathbf{C}\) contains age, sex, self-reported race/ethnicity, MESA study
site, and the first five genotype-derived ancestry principal
components. We first calculated the portion of LDL-C not explained by
these covariates,
\[
\mathbf{r}_D=\mathbf{M}_\mathbf{Q} \mathbf{d},
\]
and the portion of carotid IMT not explained by LDL-C or the
covariates,
\[
\mathbf{r}_Y=\mathbf{M}_{[\mathbf{Q},\mathbf{d}]}\mathbf{y},
\]
where \(\mathbf{M}_A\) denotes residualization with respect to the column space
of \(\mathbf A\).

The two residual vectors were transformed to standardized rank-normal
scores, denoted by \(\mathbf u_D\) and \(\mathbf u_Y\). Following this transformation,
\(\mathbf{u}_Y\) was residualized with respect to \(\mathbf{u}_D\) and restandardized.
Consequently, \(\mathbf{u}_D\) and \(\mathbf{u}_Y\) were orthogonal in the sample before
discretization, allowing the exposure-relevance and conditional
outcome-association components to be introduced through separate
residual directions.

For participant \(i\) and candidate pseudo-SNP \(j\), a latent score
was generated as
\[
S_{ij}
=
\rho_D u_{D,i}
+
\rho_Y u_{Y,i}
+
\sqrt{1-\rho_D^2-\rho_Y^2}\,e_{ij},
\qquad
e_{ij}\sim N(0,1),
\]
with
\[
\rho_D=0.30
\qquad\text{and}\qquad
\rho_Y=0.20.
\]

The latent score was discretized to obtain a genotype-like dosage
\[
z_{ij}^{\mathrm{ps}}\in\{0,1,2\}.
\]
Specifically, for target minor-allele frequency \(f=0.30\), we defined
\[
c_0=\Phi^{-1}\!\left((1-f)^2\right),
\qquad
c_1=\Phi^{-1}\!\left(1-f^2\right),
\]
and assigned
\[
z_{ij}^{\mathrm{ps}}
=
\begin{cases}
0, & S_{ij}\leq c_0,\\
1, & c_0<S_{ij}\leq c_1,\\
2, & S_{ij}>c_1.
\end{cases}
\]
These thresholds correspond to the Hardy--Weinberg target
probabilities
\[
\Pr(z_{ij}^{\mathrm{ps}}=0)=(1-f)^2,\qquad
\Pr(z_{ij}^{\mathrm{ps}}=1)=2f(1-f),\qquad
\Pr(z_{ij}^{\mathrm{ps}}=2)=f^2.
\]

Each candidate pseudo-SNP was evaluated using two partial regressions.
First, its exposure relevance was assessed using
\[
\mathbf{d}
=
\mathbf{Q}\boldsymbol{\delta}
+
\gamma_j \mathbf{z}_j^{\mathrm{ps}}
+
\boldsymbol{\nu},
\]
and \(F_j^{(\mathbf{d})}\) denoted the one-degree-of-freedom \(F\)-statistic for
testing
\[
H_0:\gamma_j=0.
\]
Second, its conditional outcome association was assessed using
\[
\mathbf{y}
=
\mathbf{Q}\boldsymbol{\eta}
+
\beta \mathbf{d}
+
\alpha_j \mathbf{z}_j^{\mathrm{ps}}
+
\boldsymbol{\varepsilon},
\]
and \(t_j^{(\mathbf{y})}\) denoted the \(t\)-statistic for testing
\[
H_0:\alpha_j=0.
\]

A candidate pseudo-SNP was accepted only when
\[
F_j^{(\mathbf{d})}\geq10,
\qquad
t_j^{(\mathbf{y})}\geq3,
\]
and
\[
\max_{\mathbf{z}_k\in\mathcal C_{j-1}}
\operatorname{Cor}
\left(
\mathbf{z}_j^{\mathrm{ps}},\mathbf{z}_k
\right)^2
\leq0.10,
\]
where \(\mathcal C_{j-1}\) contained all observed MESA SNPs and all
previously accepted pseudo-SNPs. The first criterion ensured empirical
relevance for LDL-C after covariate adjustment. The second imposed a
positive conditional association with carotid IMT after adjustment for
LDL-C and the covariates. The third prevented strong empirical
correlation with either the observed SNPs or previously accepted
pseudo-SNPs. Because the pseudo-SNPs do not represent genomic loci,
this squared-correlation criterion describes empirical dependence and
should not be interpreted as biological linkage disequilibrium.

Let
\[
\mathbf{Z}^{\mathrm{obs}}\in\mathbb{R}^{n\times L}
\]
denote the observed MESA instrument matrix and let
\[
\mathbf{Z}^{\mathrm{ps}}\in\{0,1,2\}^{n\times L_{\mathrm{ps}}}
\]
denote the matrix of accepted pseudo-SNPs. In the present experiment,
\[
L=L_{\mathrm{ps}}=8,
\]
and the augmented candidate-instrument matrix was
\[
\mathbf{Z}^{\mathrm{aug}}
=
\left[
\mathbf{Z}^{\mathrm{obs}},
\mathbf{Z}^{\mathrm{ps}}
\right]
\in\mathbb{R}^{n\times16}.
\]

Although the minimum acceptance threshold was \(F_j^{(\mathbf{d})}=10\), the
realized pseudo-SNPs were considerably stronger. Within each
replication, we calculated the mean and minimum of the eight individual
pseudo-SNP first-stage \(F\)-statistics. Across the 50 replications,
the run-level means averaged 96.91, with an SD of 7.07 and a range of
80.95 to 115.25. The run-level minimum \(F\)-statistics averaged 70.34,
with an SD of 10.65 and a range of 42.33 to 90.33. Thus, the smallest
individual pseudo-SNP first-stage \(F\)-statistic in the entire
experiment was 42.33.

These diagnostics describe the separate covariate-adjusted association
of each pseudo-SNP with LDL-C. They are not the joint first-stage
\(F\)-statistic for \(\mathbf{Z}^{\mathrm{aug}}\), nor the conditional
first-stage \(F\)-statistic for the observed SNPs after adjustment for
the pseudo-SNPs.

\paragraph{Estimation and evaluation.}

We generated 50 augmented instrument matrices using distinct random
seeds. TEAM-IV was fitted using five folds, the
1 \(\operatorname{SE}_{\mathrm{CV}}\) unioning instrument-selection rule. Covariate-adjusted sisVIVE was also fitted using five folds. For each replication, we recorded the estimated
LDL-C coefficient, the proportion of constructed pseudo-SNPs
classified as invalid, and the proportion of observed MESA SNPs
retained as valid.

The TEAM-IV coefficient was a post-selection,
partition-conditioned 2SLS estimate. For replication \(b\), let
\(\widehat V^{(b)}\) and
\({\widehat A}^{(b)}\) denote the sets classified by TEAM-IV
as valid and invalid, respectively. Conditional on this partition, the
final structural equation was
\[
\mathbf y = \theta_0\mathbf 1_n + \beta\mathbf d + \mathbf C\boldsymbol\theta + \mathbf Z_{(\widehat A^{(b)})} \boldsymbol\alpha + \boldsymbol\varepsilon
\]
with instrument matrix
\[
\left[ \mathbf 1_n,
\mathbf{C},\,
\mathbf{Z}_{\widehat A^{(b)}},\,
\mathbf{Z}_{\widehat V^{(b)}}
\right].
\]
Thus, instruments classified as invalid entered the structural
equation as included controls and also appeared in the instrument
matrix, while instruments classified as valid served as excluded
instruments for LDL-C.

Equivalently, TEAM-IV residualized \(\mathbf{y}\), \(\mathbf{d}\), and the selected valid
instruments with respect to
\[
\left[
\boldsymbol{1},\mathbf{C},
\mathbf{Z}_{\widehat A^{(b)}}
\right]
\]
before performing the final 2SLS estimation. Consequently,
\(\widehat{\beta}\) could vary across replications even when the same
eight observed MESA SNPs were consistently retained as valid. Each
replication generated a different set of pseudo-invalid instruments,
which changed the residualized outcome, exposure, and valid-instrument
matrix entering the final 2SLS regression.

For sisVIVE, an instrument was classified as invalid when the absolute
value of its estimated direct-effect coefficient exceeded
\(10^{-8}\). The pseudo-SNP detection rate was defined as the
proportion of constructed pseudo-SNPs classified as invalid. The
observed-SNP retention rate was defined as the proportion of observed
MESA SNPs classified as valid. We use the term ``retention'' rather
than ``specificity'' because the biological validity of the observed
MESA SNPs is unknown.

\paragraph{Results.}

Table~\ref{tab:pseudo_instrument_results} and
Figure~\ref{fig:pseudo_instrument_beta_hat} summarize the results.
TEAM-IV produced a mean post-selection 2SLS estimate of
\[
\widehat{\beta}=1.21,
\]
with an empirical SD of 0.46 across the 50 pseudo-instrument draws.
TEAM-IV classified 100\% of the constructed pseudo-SNPs as invalid and
retained 100\% of the observed MESA SNPs. Although the selected
partition was therefore consistent across replications, the TEAM-IV
effect estimates varied because each newly generated set of
pseudo-invalid instruments entered the final 2SLS structural equation
as a different set of included controls.

In contrast, sisVIVE produced a mean estimate of
\[
\widehat{\beta}=7.77,
\]
with an empirical SD of 0.34. It classified only 1.2\% of the
constructed pseudo-SNPs as invalid while retaining 99.5\% of the
observed MESA SNPs. Most of the deliberately outcome-associated
pseudo-instruments therefore remained in the sisVIVE instrument set,
resulting in a substantial upward shift in its estimated LDL-C
coefficient.

For comparison, the empirical analysis using only the eight observed
MESA SNPs produced
\[
\widehat{\beta}_{\mathrm{MESA}}=2.19.
\]
This empirical estimate is shown by the horizontal reference line in
Figure~\ref{fig:pseudo_instrument_beta_hat}. TEAM-IV estimates remained
substantially closer to this empirical reference than the sisVIVE
estimates, although TEAM-IV exhibited a moderate downward shift from
the reference value.

\paragraph{Interpretation and limitations.}

This experiment should be interpreted as an adversarial,
fixed-outcome stress test rather than as a conventional simulation
with known causal ground truth. The pseudo-SNPs were constructed using
the observed outcome and therefore do not represent inherited genetic
variants or possess known causal direct effects. Their classification
as ``invalid'' refers to their deliberately induced conditional
association with the observed outcome, not to a known biological
violation of the exclusion restriction.

Similarly, the empirical estimate
\(\widehat{\beta}_{\mathrm{MESA}}=2.19\) is a reference point rather
than a known true causal effect. The SDs reported across replications
measure sensitivity to repeated pseudo-instrument construction; they
are not standard errors for the empirical causal estimate.

The design was intentionally challenging in three additional respects.
First, the acceptance rule required
\(t_{\mathbf{y}\mid \mathbf{d},\mathbf{C},G_j}\geq3\), so the constructed conditional outcome
associations were all positive rather than mixed in direction.
Second, exactly 50\% of the candidate instruments were constructed as
invalid. This places the experiment at the boundary of the
majority-valid identification condition associated with sisVIVE,
which ordinarily requires fewer than 50\% of the instruments to be
invalid \citep{kang2016}. Third, the pseudo-SNPs were substantially
stronger predictors of LDL-C than the observed SNP set: the smallest
individual pseudo-SNP \(F\)-statistic was 42.33, whereas the eight
observed MESA SNPs had a joint first-stage \(F\)-statistic of 2.16.
Consequently, this experiment evaluates method behavior under strong,
directionally aligned, outcome-informed pseudo-instruments and should
not be interpreted as a strength-matched comparison with naturally
occurring pleiotropic variants.

\subsection{Semi-synthetic validation with known invalid instruments}

The adversarial pseudo-instrument analysis retained the observed
outcome but constructed outcome-informed instruments. We complemented
that analysis with a semi-synthetic experiment in which the observed
MESA exposure, genotype matrix, and covariates were retained, while the
outcome was generated under a known structural model. This design
preserved the observed distributions and dependence structure of
LDL-C, the eight MESA SNPs, and the adjustment covariates while
providing known values of the causal coefficient, direct SNP effects,
and instrument-validity classifications.

\paragraph{Data-generating mechanism.}

Let \(\mathbf{d}\) denote observed LDL-C, let
\(\mathbf{Z}^{\mathrm{obs}}\in\mathbb{R}^{n\times8}\) denote the observed MESA
SNP matrix, and let \(\mathbf{C}\) denote the covariate matrix. LDL-C and the
covariates were retained on their original scales. Before generating
the outcomes, each SNP column was standardized according to
\[
z_{ij}
=
\frac{
z_{ij}^{\mathrm{obs}}
-
\overline{z}_j^{\mathrm{obs}}
}{
s_j^{\mathrm{obs}}
},
\]
so that each column of \(\mathbf{Z}\) had sample mean zero and sample variance
one. Consequently, the exposure coefficient was expressed per
\(1~\mathrm{mg/dL}\) increase in LDL-C, whereas each direct SNP-effect
coefficient was expressed per one-SD increase in genotype dosage.

For replication \(b\), the synthetic outcome was generated as
\[
y_i^{(b)} = d_i\beta^\star + \mathbf c_i^\top\boldsymbol\theta^\star + \mathbf z_{i\cdot}^\top\boldsymbol\alpha^\star + \epsilon_i^{(b)},
\qquad
\epsilon_i^{(b)}
\stackrel{\mathrm{iid}}{\sim}
N(0,\sigma^2).
\]

The covariate-effect vector and residual noise level were calibrated
from the observed MESA outcome, denoted by \(\mathbf{y}^{\mathrm{obs}}\), before
generating the semi-synthetic outcomes. We estimated the
covariate-effect vector using the ordinary least-squares regression
\[
\mathbf{y}^{\mathrm{obs}}
=
\boldsymbol{\theta}_{0,\mathrm{OLS}}
+
\beta_{\mathrm{OLS}}\mathbf{d}
+
\mathbf{C}\boldsymbol{\theta}_{\mathrm{OLS}}
+
\mathbf{e}_{\mathrm{OLS}}.
\]
The fitted covariate coefficients were used as
\[
\boldsymbol\theta^\star
=
\widehat{\boldsymbol{\theta}}_{\mathrm{OLS}}.
\]
The fitted intercept and LDL-C coefficient were not transferred to the
outcome generator. Omitting the fitted intercept changes only the
location of the generated outcome and not the target slope
coefficients, because the fitted TEAM-IV and sisVIVE models included an
intercept.

The noise SD was calibrated separately using the conventional
all-instrument 2SLS regression
\[
\mathbf{y}^{\mathrm{obs}}
=
\boldsymbol{\theta}_{0,\mathrm{IV}}
+
\beta_{\mathrm{IV}}\mathbf{d}
+
\mathbf{C}\boldsymbol{\theta}_{\mathrm{IV}}
+
\mathbf{e}_{\mathrm{IV}},
\]
in which the eight observed MESA SNPs were used as instruments for
LDL-C and the covariates were included as exogenous regressors. We set
\[
\sigma
=
\widehat{\operatorname{SD}}
\left(
\mathbf{e}_{\mathrm{IV}}
\right)
=
195.1~\mu\mathrm{m}.
\]
The OLS covariate coefficients and IV residual SD were used only as
data-anchored calibration values for the covariate and noise
components. They were not treated as known structural parameters.

The causal coefficient was specified separately as
\[
\beta^\star=2,
\]
corresponding to a \(2~\mu\mathrm{m}\) change in synthetic carotid IMT
per \(1~\mathrm{mg/dL}\) increase in LDL-C. The first five SNPs were
assigned the same positive direct effect,
\[
\alpha_j^\star=30,
\qquad j=1,\ldots,5,
\]
while the remaining three SNPs were assigned
\[
\alpha_j^\star=0,
\qquad j=6,7,8.
\]
Because the SNP columns were standardized,
\(\alpha_j^\star=30\) represents a \(30~\mu\mathrm{m}\) direct change
in the synthetic outcome per one-SD increase in genotype dosage,
outside the pathway through LDL-C.

The known invalid and valid sets were therefore
\[
A^\star=\{1,2,3,4,5\},
\qquad
G^\star=\{6,7,8\},
\]
respectively. Thus, five of the eight candidate instruments, or
62.5\%, were invalid by construction.

We generated 50 outcomes using independent Gaussian error draws. The
observed \(\mathbf{d}\), standardized \(\mathbf{Z}\), covariate matrix \(\mathbf{C}\), causal
coefficient \(\beta^\star\), direct-effect vector \(\alpha^\star\), and
covariate-effect vector \(\theta^\star\) were held fixed across
replications. Variation across replications therefore arose from the
newly generated outcome errors and any seeded cross-validation
partitions, rather than from resampling MESA participants or
regenerating their exposure, covariates, or genotypes.

\paragraph{Estimation procedures.}

TEAM-IV and covariate-adjusted sisVIVE were fitted to every generated
dataset. TEAM-IV used five folds, and the 1 \(\operatorname{SE}_{\mathrm{CV}}\) unioning selection
rule. Adjusted sisVIVE was fitted using its
cross-validated penalty-selection procedure.

The primary TEAM-IV effect estimate was the post-selection,
partition-conditioned 2SLS coefficient. For replication \(b\), let
\(\widehat A^{(b)}\) and
\(\widehat V^{(b)}\) denote the sets classified by TEAM-IV as
invalid and valid. Conditional on this estimated partition, TEAM-IV
fitted the structural equation
\[
\mathbf{y}^{(b)}
=
\beta \mathbf{d}
+
\mathbf{C}\boldsymbol{\theta}
+
\mathbf{Z}_{\widehat A^{(b)}}\boldsymbol{\alpha}
+
\boldsymbol{\varepsilon},
\]
using
\[
\left[
\mathbf{C},\,
\mathbf{Z}_{\widehat A^{(b)}},\,
\mathbf{Z}_{\widehat V^{(b)}}
\right]
\]
as the instrument matrix. Instruments classified as invalid therefore
entered the structural equation as included controls and also appeared
in the instrument matrix, while instruments classified as valid served
as excluded instruments for LDL-C.

For sisVIVE, an instrument was classified as invalid when the absolute
magnitude of its estimated direct-effect coefficient exceeded
\(10^{-8}\). For both methods, the estimated invalid set was compared
with the known set \(A^\star\).

\paragraph{Performance measures.}

For each method, we calculated the estimation error
\[
\widehat{\beta}^{(b)}-\beta^\star,
\]
and summarized causal-effect estimation using
\[
\operatorname{Bias}
=
\frac{1}{B}
\sum_{b=1}^B
\left(
\widehat{\beta}^{(b)}-\beta^\star
\right),\quad
\operatorname{MAE}
=
\frac{1}{B}
\sum_{b=1}^B
\left|
\widehat{\beta}^{(b)}-\beta^\star
\right|, \quad
\operatorname{RMSE}
=
\sqrt{
\frac{1}{B}
\sum_{b=1}^B
\left(
\widehat{\beta}^{(b)}-\beta^\star
\right)^2
}.
\]

Instrument classification was evaluated using sensitivity,
specificity, and positive predictive value (PPV):
\[
\operatorname{Sensitivity}
=
\frac{\mathrm{TP}}{\mathrm{TP}+\mathrm{FN}},
\qquad
\operatorname{Specificity}
=
\frac{\mathrm{TN}}{\mathrm{TN}+\mathrm{FP}},
\qquad
\operatorname{PPV}
=
\frac{\mathrm{TP}}{\mathrm{TP}+\mathrm{FP}}.
\]
Here, a positive classification denotes an instrument classified as
invalid. Classification counts were pooled across the 50
replications. We also recorded estimation and classification failures.

\paragraph{Results.}

All 50 replications were completed successfully for both methods.
TEAM-IV produced a mean estimate of
\[
1.78
\]
with an empirical SD of 3.11. Its median estimate was 2.27, compared
with the true value \(\beta^\star=2\). TEAM-IV had bias
\[
-0.22,
\]
mean absolute error 2.32, and RMSE 3.09.

Covariate-adjusted sisVIVE produced a mean estimate of
\[
3.82
\]
with an empirical SD of 3.14 and a median of 3.82. Its bias was 1.82,
mean absolute error was 2.81, and RMSE was 3.60. Thus, TEAM-IV had
smaller absolute bias, MAE, and RMSE in this majority-invalid setting,
although both methods exhibited substantial replication-to-replication
variability.

TEAM-IV correctly classified all 250 invalid-instrument occurrences
and all 150 valid-instrument occurrences across the 50 replications,
yielding sensitivity, specificity, and PPV of 100\%. Adjusted sisVIVE
also identified all 250 invalid-instrument occurrences and therefore
had sensitivity of 100\%. However, it incorrectly classified 19 of the
150 valid-instrument occurrences as invalid, producing specificity of
87.3\% and PPV of 92.9\%. These false-positive classifications occurred
in 17 of the 50 replications.

\begin{table}[!htbp]
\centering
\caption{Semi-synthetic estimation and instrument-classification
performance across 50 replications.}
\label{tab:semisynthetic_results}

\small
\setlength{\tabcolsep}{3.5pt}
\renewcommand{\arraystretch}{1.20}

\begin{tabularx}{\textwidth}{
    @{}
    >{\raggedright\arraybackslash}p{0.13\textwidth}
    >{\centering\arraybackslash}p{0.09\textwidth}
    >{\centering\arraybackslash}p{0.08\textwidth}
    >{\centering\arraybackslash}p{0.08\textwidth}
    >{\centering\arraybackslash}p{0.08\textwidth}
    >{\centering\arraybackslash}p{0.08\textwidth}
    >{\centering\arraybackslash}p{0.10\textwidth}
    >{\centering\arraybackslash}p{0.10\textwidth}
    >{\centering\arraybackslash}X
    @{}
}
\toprule
\textbf{Method}
& \textbf{Mean \(\widehat{\beta}\)}
& \textbf{SD}
& \textbf{Bias}
& \textbf{MAE}
& \textbf{RMSE}
& \textbf{Sensitivity}
& \textbf{Specificity}
& \textbf{PPV} \\
\midrule

TEAM-IV
& \(1.78\)
& \(3.11\)
& \(-0.22\)
& \(2.32\)
& \(3.09\)
& \(100.0\%\)
& \(100.0\%\)
& \(100.0\%\) \\

sisVIVE
& \(3.82\)
& \(3.14\)
& \(1.82\)
& \(2.81\)
& \(3.60\)
& \(100.0\%\)
& \(87.3\%\)
& \(92.9\%\) \\

\bottomrule
\end{tabularx}

\vspace{0.5em}

\begin{minipage}{\textwidth}
\footnotesize
\textit{Note:}
The true exposure coefficient was \(\beta^\star=2\). Instruments 1--5
were assigned direct effects of \(\alpha_j^\star=30\), while
instruments 6--8 were valid. Sensitivity, specificity, and PPV treat
classification as invalid as the positive result. The observed MESA
exposure, instruments, and covariates were fixed across replications;
only the generated outcome error and seeded estimation procedures
varied.
\end{minipage}
\end{table}

\paragraph{Interpretation and limitations.}

This semi-synthetic experiment provided known causal and
instrument-validity truth while retaining the observed MESA exposure,
genotype, and covariate structure. Under the specified
majority-invalid model, TEAM-IV achieved perfect instrument
classification and smaller estimation error than adjusted sisVIVE.

The experiment nevertheless represents a deliberately challenging
setting. Five of eight instruments were invalid, violating the simple
majority-valid condition under which sisVIVE is ordinarily guaranteed
to identify the causal effect \citep{kang2016}. All five direct effects
also had the same positive magnitude, producing directional rather
than balanced pleiotropy. Consequently, the results should not be
interpreted as describing sisVIVE performance under its standard
identification conditions.

In addition, the observed MESA instruments were weak predictors of
LDL-C, with a joint first-stage \(F\)-statistic of 2.16 for the
original eight-SNP set. After selection, only three instruments were
truly valid in the data-generating model. This weak identifying
variation likely contributed to the substantial empirical SD and RMSE
of both causal-effect estimators, even though TEAM-IV classified the
instruments correctly. Finally, the 50 replications varied the
generated outcome error but did not resample participants. The results
therefore quantify performance conditional on the observed MESA design
matrix rather than repeated-sampling performance across new
populations.

The calibration combined covariate coefficients from an OLS model with
a residual SD from an IV model. Moreover, the IV noise calibration was
based on the eight observed SNPs, which had a weak joint first stage.
Consequently, \(\theta^\star\) and \(\sigma\) should be interpreted as
data-anchored calibration values rather than estimates of the true
structural covariate effects and error variance.

\subsection{Comparison of empirical and validation results}

In the empirical MESA analysis, TEAM-IV and sisVIVE produced the same estimate, \(\widehat{\beta}=2.19\), and both returned \(\widehat{\boldsymbol{\alpha}}=\boldsymbol{0}\). Thus, neither method identified evidence of direct effects among the eight observed SNPs. This agreement should not be interpreted as proof that all instruments were valid, given the weak joint first stage.

The validation analyses revealed substantial differences between the methods. In the fixed-outcome adversarial experiment, TEAM-IV identified all outcome-informed pseudo-instruments, whereas sisVIVE identified only \(1.2\%\). TEAM-IV estimates also remained substantially closer to the empirical reference estimate, although that estimate was not known causal truth. In the semi-synthetic experiment, where \(\beta^\star=2\) and the invalid-instrument set was known, TEAM-IV exhibited smaller bias, MAE, and RMSE than sisVIVE and correctly classified all valid and invalid instruments.

Taken together, these results indicate that TEAM-IV was more robust than sisVIVE to the large proportions of invalid instruments considered in these MESA-anchored experiments. This conclusion is specific to the examined designs: the pseudo-instruments were deliberately outcome-informed and unusually strong, and the invalid-instrument proportions were at or above the majority-valid boundary underlying standard sisVIVE guarantees. Accordingly, the validation results demonstrate comparative performance under challenging, controlled conditions rather than general superiority across all instrumental-variable settings.

\section{Discussion}

Across the simulation settings considered, TEAM-IV generally produced lower estimation error than sisVIVE and CIIV as the proportion of invalid instruments increased. Its advantage was most pronounced when invalid instruments constituted a majority of the candidate set and had relatively large, common-sign direct effects. TEAM-IV also remained comparatively stable under correlated instruments and heterogeneous direct-effect signs, although its advantage over sisVIVE narrowed in several mixed-sign settings. These findings support the robustness of TEAM-IV under the data-generating mechanisms examined, but they do not imply uniform superiority across all instrumental-variable settings.

In the empirical MESA application, TEAM-IV and sisVIVE produced the same estimate of the LDL-C effect and neither method selected any SNP as invalid. This agreement demonstrates that TEAM-IV can produce results consistent with an established invalid-instrument method when both procedures retain the full instrument set. 

Several limitations suggest directions for future research. The principal simulations used ($n=2000$) observations and ($L=10$) candidate instruments under a restricted collection of instrument covariance and effect-size structures. Further work should examine settings in which the number of candidate instruments is large relative to the sample size or increases with the sample size. Additional priorities include theoretical guarantees for instrument selection, valid uncertainty quantification following instrument selection, robustness to very weak instruments, and complex linkage disequilibrium.

Overall, the results show that penalized instrument-selection methods, combined with post-selection, partition-conditioned 2SLS estimation, can provide substantial robustness to direct instrument effects. TEAM-IV appears particularly promising in settings where invalid instruments are sufficiently numerous to violate the majority-valid condition required by existing methods.

\section{Appendix A: Construction of the adjusted exposure}
\label{app:dadj}

Define
\[
\mathbf P_{\mathbf Z}
=
\mathbf Z(\mathbf Z^\top\mathbf Z)^{-1}\mathbf Z^\top,
\qquad
\mathbf z_{\mathrm{sum}}=\mathbf Z\mathbf 1_L.
\]
Decompose the exposure vector into its component in $\mathrm{col}(\mathbf Z)$ and its orthogonal residual:
\[
\mathbf d
=
\mathbf d_{\mathrm{colZ}}+\mathbf d_\perp,
\qquad
\mathbf d_{\mathrm{colZ}}=\mathbf P_{\mathbf Z}\mathbf d,
\qquad
\mathbf d_\perp=(\mathbf I-\mathbf P_{\mathbf Z})\mathbf d.
\]
Define
\[
\hat\gamma_{\mathrm{sum}}
\equiv
\arg\min_{\gamma\in\mathbb R}
\|\mathbf d_{\mathrm{colZ}}-\gamma\,\mathbf z_{\mathrm{sum}}\|_2^2
=
\frac{\mathbf z_{\mathrm{sum}}^\top\mathbf d_{\mathrm{colZ}}}
{\mathbf z_{\mathrm{sum}}^\top\mathbf z_{\mathrm{sum}}},
\]
and set
\[
\mathbf d_{\mathrm{adj}}
\equiv
\mathbf d_\perp+\hat\gamma_{\mathrm{sum}}\mathbf z_{\mathrm{sum}}.
\]
Since $\mathbf z_{\mathrm{sum}}\in\mathrm{col}(\mathbf Z)$ and
$\mathbf d_\perp\perp\mathrm{col}(\mathbf Z)$, this construction satisfies
\[
\mathbf P_{\mathbf Z}\mathbf d_{\mathrm{adj}}
=
\hat\gamma_{\mathrm{sum}}\mathbf z_{\mathrm{sum}},
\qquad
(\mathbf I-\mathbf P_{\mathbf Z})\mathbf d_{\mathrm{adj}}
=
\mathbf d_\perp
=
(\mathbf I-\mathbf P_{\mathbf Z})\mathbf d.
\]
Thus, $\mathbf d_{\mathrm{adj}}$ preserves the component of $\mathbf d$ orthogonal to
$\mathrm{col}(\mathbf Z)$ while replacing the projected exposure
$\mathbf P_{\mathbf Z}\mathbf d$ by its least-squares approximation in the one-dimensional subspace
$\mathrm{span}(\mathbf Z\mathbf 1_L)$. This corresponds to a projected exposure with common relative loading across the oriented instrument columns.

In implementation, we may flip the sign of individual columns of $\mathbf Z$ to enforce a common orientation of the estimated first-stage coefficients. Such sign flips do not change $\mathrm{col}(\mathbf Z)$, and hence do not change $\mathbf P_{\mathbf Z}$, but they do determine the oriented sum direction
$\mathbf z_{\mathrm{sum}}=\mathbf Z\mathbf 1_L$.
We denote the adjusted projected exposure by
\[
\widehat{\mathbf d}_{\mathrm{hom}}
:=
\mathbf P_{\mathbf Z}\mathbf d_{\mathrm{adj}}.
\]
Since $\mathbf z_{\mathrm{sum}}\in\mathrm{col}(\mathbf Z)$ and
$\mathbf d_\perp\perp\mathrm{col}(\mathbf Z)$, we have
\[
\widehat{\mathbf d}_{\mathrm{hom}}
=
\mathbf P_{\mathbf Z}\mathbf d_{\mathrm{adj}}
=
\hat\gamma_{\mathrm{sum}}\mathbf z_{\mathrm{sum}}.
\]

A motivating example showing how heterogeneous first-stage strengths can distort MCP direct-effect fitting is given in Supplementary Note S1.

\section{Supplementary Note S1: Motivating Example for $\mathbf d_{\mathrm{adj}}$}
\label{supp:dadj-motivation}
This note gives a simple deterministic construction showing how heterogeneous first-stage strength can cause an MCP-penalized direct-effect fit to mask a truly invalid instrument.

Consider the MCP-penalized objective with \emph{per-observation} squared-error loss
\[
Q(\boldsymbol\alpha,\beta)
=\frac{1}{2n}\bigl\|\mathbf y - \mathbf Z\boldsymbol\alpha - \mathbf P_{\mathbf Z}\mathbf d\,\beta\bigr\|_2^2
\;+\;\sum_{j=1}^L p_{\lambda,\psi}(\alpha_j),
\]

where $p_{\lambda,\psi}$ is the minimax concave penalty (MCP).

Assume the data satisfy the (constant-effects) linear IV model
\[
\mathbf d = \mathbf Z\boldsymbol\gamma^\star + \boldsymbol\nu,
\qquad
\mathbf y = \mathbf d\,\beta^\star + \mathbf Z\boldsymbol\alpha^\star + \boldsymbol\epsilon,
\]
with no stochastic assumptions (all bounds below are deterministic).

Assume that only instrument \(1\) is invalid:

$$
\boldsymbol\alpha^\star=(a,0,\ldots,0)^\top,\qquad a\neq0,
$$

and that instrument \(1\) is much stronger than the others:

$$
|\gamma_1^\star|=M_1\to\infty,
\qquad
|\gamma_j^\star|\le M_2,\quad j\ge2,
$$
for fixed $M_2<\infty$.

We use two elementary properties of the MCP penalty
\begin{enumerate}
\item (Lipschitz) $\;|p_{\lambda,\psi}(u)-p_{\lambda,\psi}(v)| \le \lambda |u-v|$ for all $u,v$;
\item (Positive at nonzero) $p_{\lambda,\psi}(a)>0$ whenever $a\neq 0$.
\end{enumerate}

Define the competing parameters

\[
\delta:=\frac{a}{\gamma_1^\star},
\qquad
\beta^{(c)}:=\beta^\star+\delta,
\qquad
\boldsymbol\alpha^{(c)}
:=
\boldsymbol\alpha^\star-\delta\boldsymbol\gamma^\star.
\]

Then

\[
\alpha_1^{(c)}=a-\delta\gamma_1^\star=0.
\]

Thus the competitor treats the truly invalid first instrument as valid by offsetting its direct effect through a change in \(\beta\).
Let
\[
\mathbf r^\star
:=
\mathbf y-\mathbf Z\boldsymbol\alpha^\star-\mathbf P_{\mathbf Z}\mathbf d\,\beta^\star.
\]
Since \((\mathbf I-\mathbf P_{\mathbf Z})\mathbf Z\boldsymbol\gamma^\star=0\),
\[
\mathbf r^\star
=
\boldsymbol\epsilon+\beta^\star(\mathbf I-\mathbf P_{\mathbf Z})\boldsymbol\nu.
\]
The competitor residual is
\[
\mathbf r^{(c)}
:=
\mathbf y-\mathbf Z\boldsymbol\alpha^{(c)}
-\mathbf P_{\mathbf Z}\mathbf d\,\beta^{(c)}
=
\mathbf r^\star-\delta\mathbf P_{\mathbf Z}\boldsymbol\nu.
\]
Therefore,
\begin{align}
\frac{1}{2n}\|\mathbf r^{(c)}\|_2^2
-
\frac{1}{2n}\|\mathbf r^\star\|_2^2
&=
-\frac{\delta}{n}
\langle \mathbf r^\star,\mathbf P_{\mathbf Z}\boldsymbol\nu\rangle
+
\frac{\delta^2}{2n}
\|\mathbf P_{\mathbf Z}\boldsymbol\nu\|_2^2
\nonumber\\
&\le
\frac{|\delta|}{n}
\|\mathbf r^\star\|_2
\|\mathbf P_{\mathbf Z}\boldsymbol\nu\|_2
+
\frac{\delta^2}{2n}
\|\mathbf P_{\mathbf Z}\boldsymbol\nu\|_2^2 .
\label{eq:supp-loss-bound}
\end{align}
For the penalty part,
\begin{align}
\sum_{j=1}^L p_{\lambda,\psi}(\alpha_j^{(c)})
-
\sum_{j=1}^L p_{\lambda,\psi}(\alpha_j^\star)
&=
-p_{\lambda,\psi}(a)
+
\sum_{j=2}^L
\{p_{\lambda,\psi}(-\delta\gamma_j^\star)-p_{\lambda,\psi}(0)\}
\nonumber\\
&\le
-p_{\lambda,\psi}(a)
+
\lambda|\delta|\sum_{j=2}^L|\gamma_j^\star|
\nonumber\\
&\le
-p_{\lambda,\psi}(a)
+
\lambda|\delta|(L-1)M_2 .
\label{eq:supp-penalty-bound}
\end{align}
Combining \eqref{eq:supp-loss-bound}--\eqref{eq:supp-penalty-bound} and using
\(|\delta|=|a|/M_1\) gives
\[
Q(\boldsymbol\alpha^{(c)},\beta^{(c)})
-
Q(\boldsymbol\alpha^\star,\beta^\star)
\le
-p_{\lambda,\psi}(a)
+
O(M_1^{-1})
+
O(M_1^{-2}),
\]
for fixed \((\mathbf Z,\boldsymbol\nu,\boldsymbol\epsilon)\), fixed \(n,L,\lambda,\psi,a,M_2\), and
\(M_1=|\gamma_1^\star|\to\infty\). Hence, for sufficiently large \(M_1\),
\[
Q(\boldsymbol\alpha^{(c)},\beta^{(c)})
<
Q(\boldsymbol\alpha^\star,\beta^\star).
\]
This example shows that when one invalid instrument has a sufficiently large first-stage loading, the penalized objective prefers a fit that sets its direct-effect coefficient to zero and compensates through a small shift in \(\beta\). The adjustment \(\mathbf d_{\mathrm{adj}}\) is designed to reduce this mechanism by replacing the heterogeneous projected exposure \(\mathbf P_{\mathbf Z}\mathbf d\) with a projected exposure whose instrument-specific loadings are homogenized along \(\mathrm{span}(\mathbf Z\mathbf 1_L)\).

The same argument applies if a fixed subset of instruments \(j\ge2\) are also invalid, provided their first-stage loadings remain bounded as \(M_1\to\infty\) and their direct effects are fixed. The competitor changes their direct-effect coefficients by \(-\delta\gamma_j^\star=O(M_1^{-1})\), so the induced penalty changes are \(O(M_1^{-1})\), while the penalty decrease from setting \(\alpha_1\) to zero remains the fixed quantity \(p_{\lambda,\psi}(a)>0\).

\section{Appendix B: Ridge-ordering places valid instruments in contiguous block}
\subsection{Symmetric population regime}

This appendix studies population conditions under which the ridge-ordering step places valid instruments together in the ordering induced by the ridge direct-effect estimates. The analysis is conducted under a simplified symmetric population regime designed to isolate the behavior of the ordering step.

First, the instrument-strength coefficients are assumed to be homogeneous,
\[
\bg^\star
=
\gamma_0 \mathbf 1_L .
\]

Second, the instrument second-moment matrix is assumed to be compound symmetric,
\[
\boldsymbol\Sigma
=
\mathbb E[\z\z^\top]
=
(1-\rho)\mathbf I_L
+
\rho\,\mathbf 1_L\mathbf 1_L^\top,
\qquad
\rho\in[0,1).
\]

Third, the direct effects are assumed to take a fixed finite number of invalid levels. Specifically, let \(V^\star\) denote the valid instrument set and let
\[
A^\star=A_1\cup\cdots\cup A_m
\]
denote a partition of the invalid instrument set into \(m\) direct-effect classes, where \(m\) is fixed. For \(j\in V^\star\),
\[
\alpha_j^\star=0,
\]
while for \(j\in A_r\),
\[
\alpha_j^\star=\tau_r,
\qquad r=1,\ldots,m.
\]
The invalid direct-effect levels are assumed to be separated from zero:
\[
\Delta_\star
:=
\min_{1\le r\le m}|\tau_r|
>0.
\]

In addition, the TEAM-IV construction replaces the projected exposure by a homogenized projected exposure satisfying
\[
D_Z^{\mathrm{hom}} =
\kappa\,\mathbf z^\top \mathbf 1_L
\qquad
\text{for some } \kappa\in\mathbb R.\]
This is not an additional modeling assumption, but rather a property of the adjusted exposure construction described in Section~\ref{homexposure}.

Under these conditions, the population ridge target is shown to be constant within each direct-effect class. Moreover, the population ridge level for the valid instruments differs from the population ridge level for every invalid direct-effect class. Consequently, the population ordering induced by the ridge target places the valid instruments into a contiguous block; under standard consistency of the sample ridge estimates, the finite-sample ordering used by TEAM-IV converges to this population ordering.

\subsection{Population ridge target under a single invalid direct-effect level}

We first consider the special case in which all invalid instruments share a common direct-effect level. Thus \(m=1\), \(A^\star=A_1\), and \(\tau=\tau_1\). The extension to a fixed finite number of invalid direct-effect levels is given in the following subsection.

Let
\[
(\boldsymbol\alpha^{\mathrm{ridge},\star},
\beta^{\mathrm{ridge},\star})
=
\arg\min_{\boldsymbol\alpha,\beta}
\left\{
\frac12
\mathbb E
\!\left[
\left(
Y-\mathbf z^\top\boldsymbol\alpha
-
D_Z^{\mathrm{hom}}\beta
\right)^2
\right]
+
\lambda \|\boldsymbol\alpha\|_2^2
\right\},
\]
where \(D_Z^{\mathrm{hom}}=\kappa \mathbf{z}^\top \one_L\) is the homogenized projected exposure described in Section~\ref{homexposure}.

\begin{proposition}
Assume the symmetric population regime of Section B.1 in the special case \(m=1\). Then there exist constants \(c_V\) and \(c_A\) such that
\[
\alpha_j^{\mathrm{ridge},\star}
=
c_V,
\qquad
j\in V^\star,
\]
and
\[
\alpha_j^{\mathrm{ridge},\star}
=
c_A,
\qquad
j\in A^\star.
\]
Moreover,
\[
c_A-c_V
=
\frac{1-\rho}{1-\rho+2\lambda}\tau.
\]
Consequently, since \(\rho\in[0,1)\), \(\lambda\ge0\), and \(\tau\neq0\),
\[
|c_A-c_V|
=
\frac{1-\rho}{1-\rho+2\lambda}|\tau|
>0.
\]
\end{proposition}

In conclusion, ordering the instruments by the population ridge target places all valid instruments consecutively in the ordering. Standard consistency arguments imply that the sample ridge estimates \(\{\hat\alpha_j\}_{j=1}^L\) inherit this ordering asymptotically.

\begin{proof}
Under the symmetric regime, the population ridge objective is invariant under permutations within \(V^\star\) and within \(A^\star\). Since the ridge objective is strictly convex in \((\boldsymbol\alpha,\beta)\), its minimizer is unique. Therefore the minimizer is also invariant to within-class permutations: there exist constants \(c_V\) and \(c_A\) such that
\[
\alpha_j^{\mathrm{ridge},\star}=c_V,
\qquad j\in V^\star,
\]
and
\[
\alpha_j^{\mathrm{ridge},\star}=c_A,
\qquad j\in A^\star.
\]

The population first-order condition for \(\boldsymbol\alpha\), evaluated at the minimizer, is
\[
(\boldsymbol\Sigma+2\lambda \mathbf I_L)\boldsymbol\alpha^{\mathrm{ridge},\star}
+
\mathbf q\beta^{\mathrm{ridge},\star}
=
\mathbf m,
\]
where
\[
\mathbf q=\mathbb E[\mathbf zD_Z^{\mathrm{hom}}],
\qquad
\mathbf m=\mathbb E[\mathbf zY].
\]
By construction, there exists a scalar \(\kappa\) such that
\[
D_Z^{\mathrm{hom}}=\kappa \mathbf z^\top\mathbf 1_L.
\]
Thus
\[
\mathbf q
=
\mathbb E[\mathbf zD_Z^{\mathrm{hom}}]
=
\kappa\mathbb E[\mathbf z\mathbf z^\top]\mathbf 1_L
=
\kappa\boldsymbol\Sigma\mathbf 1_L.
\]
Since
\[
\boldsymbol\Sigma
=
(1-\rho)\mathbf I_L+\rho \mathbf 1_L\mathbf 1_L^\top,
\]
we have
\[
\boldsymbol\Sigma\mathbf 1_L
=
\{1-\rho+\rho L\}\mathbf 1_L.
\]
Therefore \(\mathbf q\beta^{\mathrm{ridge},\star}\) contributes the same value to every coordinate.

Next, using
\[
Y=\beta^\star D+\mathbf z^\top\boldsymbol\alpha^\star+\epsilon,
\qquad
\mathbb E[\mathbf z\epsilon]=\zero_L,
\]
we obtain
\[
\mathbf m
=
\mathbb E[\mathbf zY]
=
\beta^\star \mathbb E[\mathbf zD]
+
\mathbb E[\mathbf z\mathbf z^\top]\boldsymbol\alpha^\star.
\]
Since
\[
\mathbb E[\mathbf zD]
=
\boldsymbol\Sigma\boldsymbol\gamma^\star
\]
and
\[
\boldsymbol\gamma^\star=\gamma_0\mathbf 1_L,
\]
the first term is proportional to \(\mathbf 1_L\). Hence, when subtracting a valid-coordinate first-order condition from an invalid-coordinate first-order condition, the common exposure and first-stage terms cancel.

Let \(j\in A^\star\) and \(k\in V^\star\). For any vector \(\mathbf u \in \mathbb{R}^L\),
\[
[(\boldsymbol\Sigma+2\lambda \mathbf I_L)\mathbf u]_\ell
=
(1-\rho)u_\ell
+
\rho\,\mathbf 1_L^\top\mathbf u
+
2\lambda u_\ell.
\]
Applying this with
\(\mathbf u=\boldsymbol\alpha^{\mathrm{ridge},\star}\), we get
\[
\left[
(\boldsymbol\Sigma+2\lambda \mathbf I_L)
\boldsymbol\alpha^{\mathrm{ridge},\star}
\right]_j
-
\left[
(\boldsymbol\Sigma+2\lambda \mathbf I_L)
\boldsymbol\alpha^{\mathrm{ridge},\star}
\right]_k
=
(1-\rho+2\lambda)(c_A-c_V),
\]
because the common term
\(\rho\,\mathbf 1_L^\top\boldsymbol\alpha^{\mathrm{ridge},\star}\)
cancels.

It remains to compute the corresponding right-hand side difference. Since
\[
\boldsymbol\alpha^\star=\tau\mathbf 1_{A^\star},
\]
the only non-common part of \(\mathbf m\) is
\[
\boldsymbol\Sigma\boldsymbol\alpha^\star
=
\boldsymbol\Sigma(\tau\mathbf 1_{A^\star}).
\]
Thus
\[
m_j-m_k
=
\left[
\boldsymbol\Sigma(\tau\mathbf 1_{A^\star})
\right]_j
-
\left[
\boldsymbol\Sigma(\tau\mathbf 1_{A^\star})
\right]_k.
\]
Because \(j\in A^\star\) and \(k\in V^\star\),
\[
\left[
\boldsymbol\Sigma(\tau\mathbf 1_{A^\star})
\right]_j
=
\tau(1-\rho+\rho |A^\star|),
\qquad
\left[
\boldsymbol\Sigma(\tau\mathbf 1_{A^\star})
\right]_k
=
\tau\rho |A^\star|.
\]
Therefore
\[
m_j-m_k
=
(1-\rho)\tau.
\]

Equating the coordinate differences in the first-order condition gives
\[
(1-\rho+2\lambda)(c_A-c_V)
=
(1-\rho)\tau.
\]
Hence
\[
c_A-c_V
=
\frac{1-\rho}{1-\rho+2\lambda}\tau.
\]
Since \(\rho\in[0,1)\), \(\lambda\ge0\), and \(\tau\neq0\),
\[
|c_A-c_V|
=
\frac{1-\rho}{1-\rho+2\lambda}|\tau|
>
0.
\]
\end{proof}
\subsection{Extension to finitely many invalid direct-effect levels}

The preceding separation calculation extends directly to any fixed finite number of invalid direct-effect levels. Suppose that
\[
A^\star=A_1\cup\cdots\cup A_m,
\]
where \(m\) is fixed, and
\[
\alpha_j^\star=\tau_r
\qquad
\text{for } j\in A_r,\quad r=1,\ldots,m.
\]
Assume the invalid levels are separated from the valid level \(0\):
\[
\Delta_\star
:=
\min_{1\le r\le m}|\tau_r|
>0.
\]

Under the symmetric population regime, let \(k\in V^\star\) and \(j\in A_r\). The population first-order condition for
\(\boldsymbol\alpha^{\mathrm{ridge},\star}\) implies
\[
\left[
(\boldsymbol\Sigma+2\lambda \mathbf I_L)
\boldsymbol\alpha^{\mathrm{ridge},\star}
\right]_j
-
\left[
(\boldsymbol\Sigma+2\lambda \mathbf I_L)
\boldsymbol\alpha^{\mathrm{ridge},\star}
\right]_k
=
m_j-m_k,
\]
after cancellation of the common exposure and first-stage terms. Since
\[
\boldsymbol\Sigma
=
(1-\rho)\mathbf I_L+\rho \mathbf 1_L\mathbf 1_L^\top,
\]
the common \(\rho \mathbf 1_L\mathbf 1_L^\top\) contribution cancels across coordinates, giving
\[
\left[
(\boldsymbol\Sigma+2\lambda \mathbf I_L)
\boldsymbol\alpha^{\mathrm{ridge},\star}
\right]_j
-
\left[
(\boldsymbol\Sigma+2\lambda \mathbf I_L)
\boldsymbol\alpha^{\mathrm{ridge},\star}
\right]_k
=
(1-\rho+2\lambda)
\left(
\alpha_j^{\mathrm{ridge},\star}
-
\alpha_k^{\mathrm{ridge},\star}
\right).
\]
Similarly, because \(\alpha_k^\star=0\) and \(\alpha_j^\star=\tau_r\),
\[
m_j-m_k
=
(1-\rho)(\alpha_j^\star-\alpha_k^\star)
=
(1-\rho)\tau_r.
\]
Therefore
\[
(1-\rho+2\lambda)
\left(
\alpha_j^{\mathrm{ridge},\star}
-
\alpha_k^{\mathrm{ridge},\star}
\right)
=
(1-\rho)\tau_r,
\]
so
\[
\alpha_j^{\mathrm{ridge},\star}
-
\alpha_k^{\mathrm{ridge},\star}
=
\frac{1-\rho}{1-\rho+2\lambda}\tau_r.
\]

Writing \(c_V\) for the common valid-class ridge level and \(c_{A,r}\) for the ridge level of invalid class \(A_r\), we obtain
\[
c_{A,r}-c_V
=
\frac{1-\rho}{1-\rho+2\lambda}\tau_r.
\]
Hence
\[
\min_{j\in A^\star}
\left|
\alpha_j^{\mathrm{ridge},\star}
-
c_V
\right|
=
\min_{1\le r\le m}|c_{A,r}-c_V|
=
\frac{1-\rho}{1-\rho+2\lambda}
\min_{1\le r\le m}|\tau_r|.
\]
Since \(\rho\in[0,1)\), \(\lambda\ge 0\), and \(\Delta_\star>0\),
\[
\min_{j\in A^\star}
\left|
\alpha_j^{\mathrm{ridge},\star}
-
c_V
\right|
=
\frac{1-\rho}{1-\rho+2\lambda}
\Delta_\star
>0.
\]
Thus, for any fixed finite number of invalid direct-effect levels, a structural gap away from zero induces a positive population ridge-level gap between valid and invalid instruments.

\section{Window-level sign coherence of valid instruments}
\label{app:window-sign-coherence}

This appendix provides a theoretical justification for aggregating
sign-with-tolerance patterns across overlapping window-specific MCP fits.
The central result shows that the estimated direct-effect coefficients of
truly valid instruments appearing in the same window asymptotically receive
a common sign-with-tolerance classification. This conclusion permits the
valid coefficients to be jointly zero or jointly nonzero, thereby allowing
for cluster-induced shifts of the MCP solution.

\subsection{Window-specific objective and common-shift path}

Fix a window \(W\in\mathcal W\) of size \(w:=|W|\). Define the
window-specific MCP objective by
\[
Q_{n,W}(\boldsymbol\alpha_W,\beta)
:=
\frac{1}{2n}
\left\|
\mathbf y
-
\mathbf Z_W\boldsymbol\alpha_W
-
\mathbf P_{\mathbf Z_W}
\widehat{\mathbf d}_{\mathrm{hom}}\beta
\right\|_2^2
+
\sum_{j\in W}
p_{\lambda_n,\psi}(\alpha_{W,j}).
\]
\subsection{Common-shift geometry and localization conditions}
\label{app:window-sign-assumptions}

Fix a window \(W\in\mathcal W\) of size $w:=|W|$, and let
\[
\mathbf Z_W\in\mathbb R^{n\times w}
\]
denote the corresponding instrument submatrix. Define the window-specific
adjusted first-stage coefficient vector by
\[
\widehat{\boldsymbol\gamma}_W
:=
\left(
\mathbf Z_W^\top\mathbf Z_W
\right)^{-1}
\mathbf Z_W^\top
\widehat{\mathbf d}_{\mathrm{hom}},
\]
so that
\[
\mathbf P_{\mathbf Z_W}
\widehat{\mathbf d}_{\mathrm{hom}}
=
\mathbf Z_W\widehat{\boldsymbol\gamma}_W.
\]
\begin{assumption}[Window-level homogeneous adjusted first stage]
\label{ass:window-homogeneous-first-stage}
For the fixed window \(W\), there exists a nonzero scalar \(c_W\) such that
\[
\widehat{\boldsymbol\gamma}_W
=
c_W^{-1}\mathbf1_w.
\]
Equivalently,
\[
c_W\widehat{\boldsymbol\gamma}_W
=
\mathbf1_w.
\]
\end{assumption}

Let \(\mathcal T\subset\mathbb R\) be a compact set of candidate shift
levels. For each \(t\in\mathcal T\), define the window-specific common-shift
path by
\[
\boldsymbol\alpha_W^{(t)}
=
\boldsymbol\alpha_W^\star-t\mathbf 1_w,
\qquad
\beta_W^{(t)}
=
\beta^\star+c_Wt.
\]
Define the corresponding residual by
\[
\mathbf r_W^{(t)}
:=
\mathbf y
-
\mathbf Z_W\boldsymbol\alpha_W^{(t)}
-
\mathbf P_{\mathbf Z_W}
\widehat{\mathbf d}_{\mathrm{hom}}\,
\beta_W^{(t)}.
\]

The fitted value is constant along this path. Indeed,
\[
\begin{aligned}
&
\mathbf Z_W\boldsymbol\alpha_W^{(t)}
+
\mathbf P_{\mathbf Z_W}
\widehat{\mathbf d}_{\mathrm{hom}}\,
\beta_W^{(t)}
\\
&=
\mathbf Z_W
\left(
\boldsymbol\alpha_W^\star-t\mathbf 1_w
\right)
+
\mathbf Z_W\widehat{\boldsymbol\gamma}_W
\left(
\beta^\star+c_Wt
\right)
\\
&=
\mathbf Z_W\boldsymbol\alpha_W^\star
+
\mathbf P_{\mathbf Z_W}
\widehat{\mathbf d}_{\mathrm{hom}}\,
\beta^\star
+
t\mathbf Z_W
\left(
-\mathbf 1_w
+
c_W\widehat{\boldsymbol\gamma}_W
\right)
\\
&=
\mathbf Z_W\boldsymbol\alpha_W^\star
+
\mathbf P_{\mathbf Z_W}
\widehat{\mathbf d}_{\mathrm{hom}}\,
\beta^\star.
\end{aligned}
\]
Consequently,
\[
\mathbf r_W^{(t)}
=
\mathbf r_W^{(0)}
\qquad
\text{for every }t\in\mathcal T.
\]

\begin{assumption}[Off-path localization conditions]
\label{ass:window-sign-regularity-1}
The following conditions hold for the fixed window \(W\).

\begin{enumerate}
\item[(i)]
The window dimension \(w\) is fixed as \(n\to\infty\).

\item[(ii)]
There exists a fixed constant \(\kappa_{Z,W}>0\) such that
\[
\Pr\left\{
\lambda_{\min}\left(
\frac{1}{n}\mathbf Z_W^\top\mathbf Z_W
\right)
\ge
\kappa_{Z,W}
\right\}
\longrightarrow1.
\]

\item[(iii)]
The empirical window-specific design--residual score has the usual
root-\(n\) order:
\[
\left\|
\frac{1}{n}
\mathbf Z_W^\top\mathbf r_W^{(0)}
\right\|_2
=
O_p(n^{-1/2}).
\]
Because
\[
\mathbf r_W^{(t)}=\mathbf r_W^{(0)}
\qquad
\text{for every }t\in\mathcal T,
\]
this is equivalent to
\[
\sup_{t\in\mathcal T}
\left\|
\frac{1}{n}
\mathbf Z_W^\top\mathbf r_W^{(t)}
\right\|_2
=
O_p(n^{-1/2}).
\]

\item[(iv)]
The MCP tuning sequence satisfies
\[
\lambda_n\longrightarrow0,
\qquad
n\lambda_n^2\longrightarrow\infty,
\]
and the MCP concavity parameter \(\psi>1\) is fixed.
\end{enumerate}
\end{assumption}

\begin{lemma}[Exclusion of off-path deviations larger than the MCP scale]
\label{lem:offpath-localization}

Suppose Assumptions~\ref{ass:window-homogeneous-first-stage}
and~\ref{ass:window-sign-regularity-1} hold. For
\(t\in\mathcal T\), consider a perturbation
\[
\beta
=
\beta_W^{(t)}+\Delta_\beta,
\qquad
\boldsymbol\alpha_W
=
\boldsymbol\alpha_W^{(t)}+\boldsymbol\Delta_W,
\]
and define the identifiable off-path coordinate
\[
\mathbf h_W
:=
\boldsymbol\Delta_W
+
\widehat{\boldsymbol\gamma}_W\Delta_\beta.
\]
Then
\[
\mathbf Z_W\boldsymbol\Delta_W
+
\mathbf P_{\mathbf Z_W}
\widehat{\mathbf d}_{\mathrm{hom}}\,
\Delta_\beta
=
\mathbf Z_W\mathbf h_W.
\]

Let \(M>0\) satisfy
\[
M
>
\sqrt{
\frac{2w\psi}{\kappa_{Z,W}}
}.
\]
Then
\[
\Pr\left[
\inf_{t\in\mathcal T}
\inf_{\substack{
\boldsymbol\Delta_W\in\mathbb R^w,\,
\Delta_\beta\in\mathbb R:\\
\|
\boldsymbol\Delta_W
+
\widehat{\boldsymbol\gamma}_W\Delta_\beta
\|_2
\ge
M\lambda_n
}}
\left\{
Q_{n,W}\left(
\boldsymbol\alpha_W^{(t)}+\boldsymbol\Delta_W,\,
\beta_W^{(t)}+\Delta_\beta
\right)
-
Q_{n,W}\left(
\boldsymbol\alpha_W^{(t)},
\beta_W^{(t)}
\right)
\right\}
>0
\right]
\longrightarrow1.
\]

Thus, uniformly over the candidate shift levels, every identifiable
off-path perturbation of magnitude at least \(M\lambda_n\) has larger
window-specific MCP objective than the corresponding point on the
common-shift path.
\end{lemma}

\begin{proof}
Fix \(t\in\mathcal T\), and consider a perturbation \[
\beta
=
\beta_W^{(t)}+\Delta_\beta,
\qquad
\boldsymbol\alpha_W
=
\boldsymbol\alpha_W^{(t)}+\boldsymbol\Delta_W.
\]
The residual at the perturbed parameter value is
\[
\begin{aligned}
\mathbf y
-
\mathbf Z_W
\bigl(
\boldsymbol\alpha_W^{(t)}+\boldsymbol\Delta_W
\bigr)
-
\mathbf P_{\mathbf Z_W}\widehat{\mathbf d}_{\mathrm{hom}}
\bigl(
\beta_W^{(t)}+\Delta_\beta
\bigr)
&=
\mathbf r_W^{(t)}
-
\mathbf Z_W\boldsymbol\Delta_W
-
\mathbf P_{\mathbf Z_W}\widehat{\mathbf d}_{\mathrm{hom}}\,\Delta_\beta
\\
&=
\mathbf r_W^{(t)}-\mathbf Z_W\mathbf h_W.
\end{aligned}
\]

Expanding the squared-error part of the objective therefore gives
\[
\begin{aligned}
&
\frac{1}{2n}
\left\|
\mathbf r_W^{(t)}-\mathbf Z_W\mathbf h_W
\right\|_2^2
-
\frac{1}{2n}
\left\|
\mathbf r_W^{(t)}
\right\|_2^2
\\
&\qquad
=
\frac{1}{2n}
\|\mathbf Z_W\mathbf h_W\|_2^2
-
\frac{1}{n}
\left\langle
\mathbf r_W^{(t)},\mathbf Z_W\mathbf h_W
\right\rangle
\\
&\qquad
=
\frac{1}{2n}
\|\mathbf Z_W\mathbf h_W\|_2^2
-
\mathbf h_W^\top
\frac{\mathbf Z_W^\top\mathbf r_W^{(t)}}{n}.
\end{aligned}
\]

Consequently,
\[
\begin{aligned}
&
Q_{n,W}\left(
\boldsymbol\alpha_W^{(t)}+\boldsymbol\Delta_W,\,
\beta_W^{(t)}+\Delta_\beta
\right)
-
Q_{n,W}\left(
\boldsymbol\alpha_W^{(t)},\beta_W^{(t)}
\right)
\\
&\quad
=
\frac{1}{2n}
\|\mathbf Z_W\mathbf h_W\|_2^2
-
\mathbf h_W^\top
\frac{\mathbf Z_W^\top\mathbf r_W^{(t)}}{n}
\\
&\qquad\quad
+
\sum_{j\in W}
\left[
p_{\lambda_n,\psi}
\left(
\alpha_{W,j}^{(t)}+\Delta_{W,j}
\right)
-
p_{\lambda_n,\psi}
\left(
\alpha_{W,j}^{(t)}
\right)
\right].
\end{aligned}
\]

By the lower-eigenvalue condition,
\[
\frac{1}{2n}
\|\mathbf Z_W\mathbf h_W\|_2^2
\ge
\frac{\kappa_{Z,W}}{2}
\|\mathbf h_W\|_2^2
\]
with probability tending to one.

Next, by Cauchy--Schwarz,
\[
-
\mathbf h_W^\top
\frac{\mathbf Z_W^\top\mathbf r_W^{(t)}}{n}
\ge
-
\|\mathbf h_W\|_2
\left\|
\frac{\mathbf Z_W^\top\mathbf r_W^{(t)}}{n}
\right\|_2.
\]

The MCP penalty is nonnegative and bounded above by
\(\psi\lambda_n^2/2\). Hence, for every coordinate,
\[
p_{\lambda_n,\psi}
\left(
\alpha_{W,j}^{(t)}+\Delta_{W,j}
\right)
-
p_{\lambda_n,\psi}
\left(
\alpha_{W,j}^{(t)}
\right)
\ge
-\frac{\psi}{2}\lambda_n^2.
\]
Summing over the \(w=|W|\) coordinates gives
\[
\sum_{j\in W}
\left[
p_{\lambda_n,\psi}
\left(
\alpha_{W,j}^{(t)}+\Delta_{W,j}
\right)
-
p_{\lambda_n,\psi}
\left(
\alpha_{W,j}^{(t)}
\right)
\right]
\ge
-\frac{w\psi}{2}\lambda_n^2.
\]

Combining these inequalities gives
\[
\begin{aligned}
&
Q_{n,W}\bigl(
\boldsymbol\alpha_W^{(t)}+\boldsymbol\Delta_W,\,
\beta_W^{(t)}+\Delta_\beta
\bigr)
-
Q_{n,W}\bigl(
\boldsymbol\alpha_W^{(t)},\beta_W^{(t)}
\bigr)
\\
&\qquad
\ge
\frac{\kappa_{Z,W}}{2}
\|\mathbf h_W\|_2^2
-
\|\mathbf h_W\|_2
\left\|
\frac{\mathbf Z_W^\top\mathbf r_W^{(t)}}{n}
\right\|_2
-
\frac{w\psi}{2}\lambda_n^2.
\end{aligned}
\]

Define
\[
B_{n,W}
:=
\sup_{t\in\mathcal T}
\left\|
\frac{\mathbf Z_W^\top\mathbf r_W^{(t)}}{n}
\right\|_2.
\]
By Assumption~\ref{ass:window-sign-regularity-1}(iii),
\[
B_{n,W}=O_p(n^{-1/2}).
\]

Because
\[
n\lambda_n^2\longrightarrow\infty,
\]
we have
\[
\frac{n^{-1/2}}{\lambda_n}
=
\frac{1}{\sqrt n\,\lambda_n}
\longrightarrow0,
\]
and therefore
\[
\frac{B_{n,W}}{\lambda_n}=o_p(1).
\]

It follows that
\[
\Pr\left\{
B_{n,W}
\le
\frac{\kappa_{Z,W}M}{4}\lambda_n
\right\}
\longrightarrow1.
\]
On this event, for every \(\mathbf h_W\) satisfying
\[
\|\mathbf h_W\|_2\ge M\lambda_n,
\]
we have
\[
B_{n,W}
\le
\frac{\kappa_{Z,W}M}{4}\lambda_n
\le
\frac{\kappa_{Z,W}}{4}\|\mathbf h_W\|_2.
\]
Hence
\[
\|\mathbf h_W\|_2B_{n,W}
\le
\frac{\kappa_{Z,W}}{4}\|\mathbf h_W\|_2^2.
\]

Substituting this bound into the objective inequality gives
\[
\begin{aligned}
&
Q_{n,W}\bigl(
\boldsymbol\alpha_W^{(t)}+\boldsymbol\Delta_W,\,
\beta_W^{(t)}+\Delta_\beta
\bigr)
-
Q_{n,W}\bigl(
\boldsymbol\alpha_W^{(t)},\beta_W^{(t)}
\bigr)
\\
&\qquad
\ge
\frac{\kappa_{Z,W}}{4}\|\mathbf h_W\|_2^2
-
\frac{w\psi}{2}\lambda_n^2
\\
&\qquad
\ge
\left(
\frac{\kappa_{Z,W}M^2}{4}
-
\frac{w\psi}{2}
\right)
\lambda_n^2.
\end{aligned}
\]

By the assumed choice of \(M\),
\[
\frac{\kappa_{Z,W}M^2}{4}
-
\frac{w\psi}{2}
>0.
\]
The lower bound is therefore strictly positive uniformly over
\(t\in\mathcal T\) and all perturbations satisfying
\[
\left\|
\boldsymbol\Delta_W+
\widehat{\boldsymbol\gamma}_W \Delta_\beta
\right\|_2
\ge M\lambda_n
\]
on an event whose probability tends to one. This proves the result.
\end{proof}
\begin{assumption}[Valid-set and sign-separation conditions]
\label{ass:window-sign-regularity-2}
\begin{enumerate}
\item[(vi)]
Let
\[
G_W^\star
:=
W\cap G^\star
\]
denote the truly valid indices contained in window \(W\). The fixed window
contains at least one truly valid instrument:
\[
G_W^\star\ne\emptyset.
\]

\item[(vii)]
Let
\[
\left(
\widehat{\boldsymbol\alpha}_W,
\widehat\beta_W
\right)
\]
be a measurable global minimizer of \(Q_{n,W}\), and define its induced
shift level by
\[
\widehat t_W
:=
\frac{\widehat\beta_W-\beta^\star}{c_W}.
\]
The selected shift lies in the candidate shift set with probability
tending to one:
\[
\Pr\left(
\widehat t_W\in\mathcal T
\right)
\longrightarrow1.
\]

\item[(viii)]
For a fixed tolerance \(\alpha_{\mathrm{tol}}>0\), the selected shift is
asymptotically separated from the sign-classification boundaries. There
exists a constant \(\delta>0\) such that
\[
\Pr\left\{
\left|
|\widehat t_W|-\alpha_{\mathrm{tol}}
\right|
>\delta
\right\}
\longrightarrow1.
\]
\end{enumerate}

\end{assumption}
\begin{theorem}[Window-level sign coherence of truly valid indices]
\label{thm:window-valid-sign-coherence}

Suppose Assumptions~\ref{ass:window-homogeneous-first-stage},
\ref{ass:window-sign-regularity-1}, and
\ref{ass:window-sign-regularity-2} hold. Fix
\[
M>
\sqrt{\frac{2w\psi}{\kappa_{Z,W}}}.
\]

Define
\[
\operatorname{sign}_{\alpha_{\mathrm{tol}}}(x)
=
\begin{cases}
-1, & x<-\alpha_{\mathrm{tol}},\\
0, & |x|\le\alpha_{\mathrm{tol}},\\
1, & x>\alpha_{\mathrm{tol}},
\end{cases}
\]
and
\[
s^\star(W)
:=
\operatorname{sign}_{\alpha_{\mathrm{tol}}}
\left(-\widehat t_W\right).
\]
Then
\[
\Pr\left\{
\operatorname{sign}_{\alpha_{\mathrm{tol}}}
\left(\widehat\alpha_{W,j}\right)
=
s^\star(W)
\text{ for every }j\in G_W^\star
\right\}
\longrightarrow1.
\]
Equivalently, for every pair
\(j,k\in G^\star_W\),
\[
\Pr\left\{
\operatorname{sign}_{\alpha_{\mathrm{tol}}}
(\widehat\alpha_{W,j})
\neq
\operatorname{sign}_{\alpha_{\mathrm{tol}}}
(\widehat\alpha_{W,k})
\right\}
\longrightarrow0.
\]
Thus, within any fixed window, all truly valid indices asymptotically receive
the same sign-with-tolerance classification, regardless of whether the
selected MCP solution remains near the oracle configuration or shifts
toward an invalid-instrument cluster.
\end{theorem}

\begin{proof}
Because
\[
\widehat\beta_W
=
\beta_W^{(\widehat t_W)},
\]
the global minimizer can be represented relative to the corresponding
shift-path point as
\[
\widehat\beta_W
=
\beta_W^{(\widehat t_W)}+\Delta_\beta
\]
with
\[
\Delta_\beta=0.
\]
Consequently, the identifiable off-path coordinate from
Lemma~\ref{lem:offpath-localization} reduces to
\[
\mathbf h_W
=
\widehat{\boldsymbol\alpha}_W
-
\boldsymbol\alpha_W^{(\widehat t_W)}.
\]

The off-path localization lemma therefore gives
\[
\Pr\left\{
\left\|
\widehat{\boldsymbol\alpha}_W
-
\boldsymbol\alpha_W^{(\widehat t_W)}
\right\|_2
<
M\lambda_n
\right\}
\longrightarrow1.
\]
Since the coordinatewise norm is bounded by the Euclidean norm,
\[
\max_{j\in W}
\left|
\widehat\alpha_{W,j}
-
\alpha_{W,j}^{(\widehat t_W)}
\right|
<
M\lambda_n
\]
with probability tending to one.

For every truly valid index \(j\in  G^\star_W\),
\[
\alpha_j^\star=0,
\]
and hence
\[
\alpha_{W,j}^{(\widehat t_W)}
=
\alpha_j^\star-\widehat t_W
=
-\widehat t_W.
\]
It follows that
\[
\max_{j\in  G^\star_W}
\left|
\widehat\alpha_{W,j}
+
\widehat t_W
\right|
<
M\lambda_n
\]
with probability tending to one.

Because \(M\) is fixed and \(\lambda_n\to0\), eventually
\[
M\lambda_n<\delta.
\]
On the event
\[
\left|
|\widehat t_W|-\alpha_{\mathrm{tol}}
\right|
>\delta,
\]
the common value \(-\widehat t_W\) lies more than \(\delta\) from the
nearest sign-classification boundary.
Therefore, every number within \(M\lambda_n\) of
\(-\widehat t_W\) has the same sign-with-tolerance classification as
\(-\widehat t_W\).

Hence,
\[
\operatorname{sign}_{\alpha_{\mathrm{tol}}}
(\widehat\alpha_{W,j})
=
\operatorname{sign}_{\alpha_{\mathrm{tol}}}
(-\widehat t_W)
=
s^\star (W)
\]
simultaneously for every
\(j\in  G^\star_W\), with probability tending to one.
The pairwise statement follows immediately.
\end{proof}

\subsection{Implications for aggregation across windows}

Lemma~\ref{lem:offpath-localization} shows that a global minimizer lies,
with probability tending to one, within an \(O(\lambda_n)\) neighborhood
of the common-shift path. Theorem~\ref{thm:window-valid-sign-coherence}
then implies that all truly valid instruments appearing in a fixed window
receive the same sign-with-tolerance classification asymptotically.

Accordingly, a window may classify its valid instruments jointly as zero,
jointly as positive, or jointly as negative. In each case, the valid
instruments avoid zero--nonzero and opposite-sign discordance within that
window. Aggregating these patterns over overlapping windows therefore
provides evidence of their joint validity even when cluster-induced shifts
move the estimated valid coefficients away from zero.

\section{Additional Simulation Results}\label{supp:add-sim}

This section reports additional simulation results that supplement the main findings in Section~\ref{simulation}. Unless otherwise stated, the data-generating mechanism, sample size, number of candidate instruments, causal effect, TEAM-IV implementation, error distribution, and performance metric are the same as in the main simulation study. In particular, we use \(n=2000\), \(L=10\), \(\beta^\star=1\), and evaluate performance using the absolute estimation error \( |\widehat\beta-\beta^\star| \) over 1000 Monte Carlo replications per setting.

The additional simulations examine departures from the main design along three dimensions. First, we consider correlated candidate instruments using compound-symmetric and autoregressive covariance structures. Second, we consider mixed-sign invalid direct effects, where the nonzero entries of \(\boldsymbol\alpha^\star\) are evenly split between positive and negative signs. Third, we consider signed and heterogeneous first-stage coefficients, where the entries of \(\boldsymbol\gamma^\star\) are no longer constrained to be positive and similarly sized. These settings assess whether the relative performance of TEAM-IV is robust to correlation among instruments and to violations of the sign-coherence patterns used in the main simulation design.

\subsection{Mixed-sign invalid direct effects with autoregressive instrument correlation}
\begin{figure}[h!]
\begin{center}
\includegraphics[width=\textwidth]{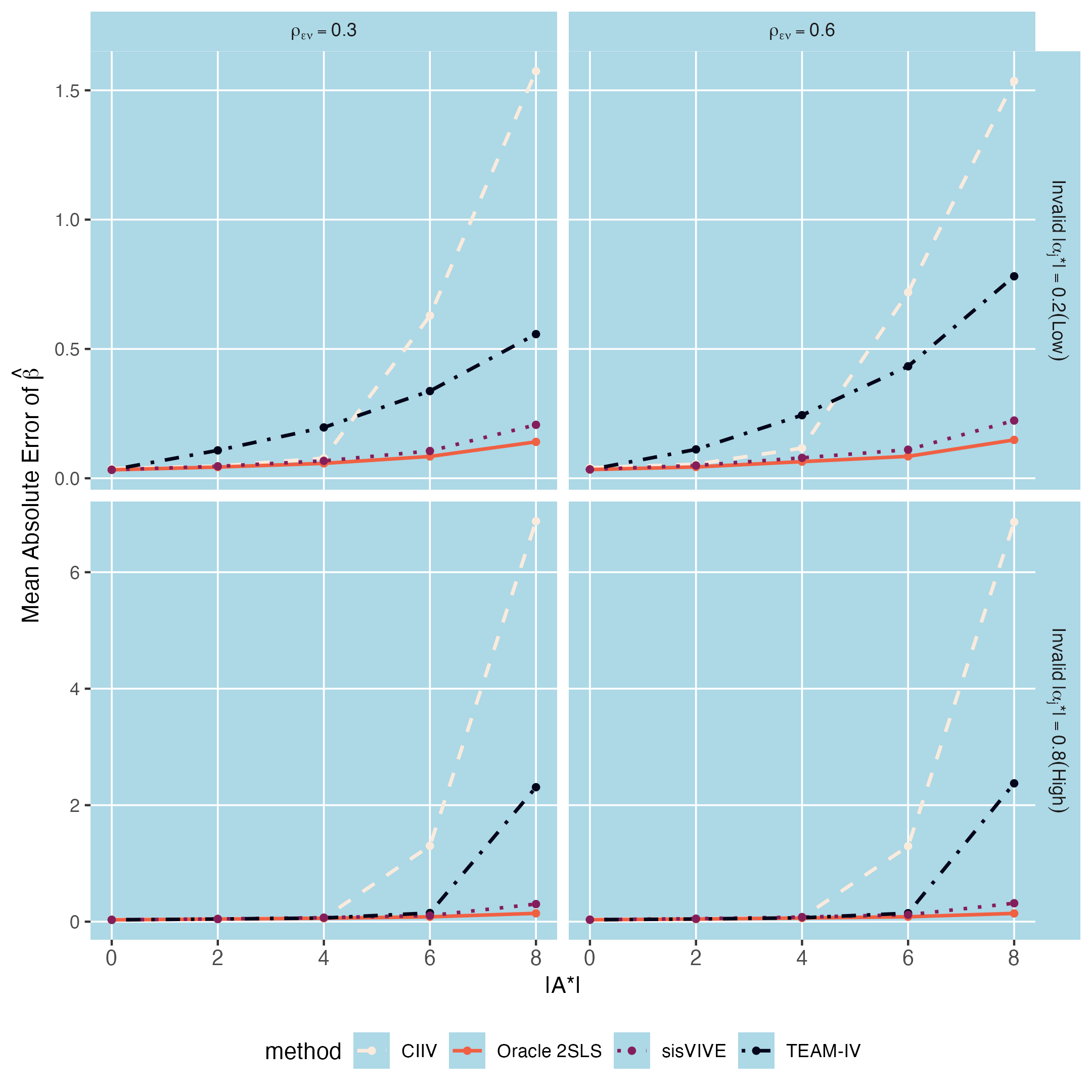}
\caption{Mean absolute error (MAE) of $\widehat\beta$ under the simulation design with autoregressive instrument covariance and invalid direct effects split evenly between positive and negative signs. Results are shown by the number of invalid instruments \(|A^\star|\), the endogeneity parameter \(\rho_{\epsilon\nu}\), and the invalid direct-effect magnitude \(a_\alpha\). Each point represents the average across 1000 Monte Carlo replications.}
\label{fig:alpmix-ar1-mae}
\end{center}
\end{figure}

Figure~\ref{fig:alpmix-ar1-mae} reports the mean absolute estimation error of $\widehat\beta$ when the candidate instruments have an autoregressive covariance structure and invalid direct effects are split evenly between positive and negative signs. This design departs from the main design in two dimensions: correlation structure of $\mathbf z_{i \cdot}$ and the sign incoherence of $\boldsymbol \alpha^\star$.

The aim of this crossed sensitivity analysis is to determine whether TEAM-IV performance degrades when these two departures occur in tandem.

The principal effect on TEAM-IV of using mixed-sign invalid direct effects simultaneously with autoregressive instrument covariance is worse performance when \(|A^\star|=8\) or \(a_\alpha=0.2\). As occurred under mixed-sign invalid direct effects with independent instruments, sisVIVE performs substantially better than under the main design.

When \(a_\alpha=0.2\), TEAM-IV performs worse than sisVIVE and oracle 2SLS across all nonzero levels of \(|A^\star|\), and worse than CIIV when \(|A^\star|\in\{2,4\}\). When \(|A^\star|=8\), CIIV error increases dramatically, coinciding with violation of the plurality rule when only two instruments are valid.

When $a_\alpha=0.8$, TEAM-IV error closely tracks that of oracle 2SLS for all $|A^\star| \leq 6$, but degrades moderately when $|A^\star| =8$. sisVIVE performs close to oracle 2SLS across all levels of $|A^\star|$, as it does in the mixed-sign invalid direct effect design under independent instruments. CIIV error again increases markedly when $|A^\star| =8$.

Overall, this two-departure design indicates a fair amount of robustness of TEAM-IV to combining autoregressive correlation in \(\mathbf Z\) with sign incoherence of \(\alpha_j^\star\) for \(j\in A^\star\). CIIV shows unexpectedly elevated error when \(|A^\star|\in\{6,8\}\). In the \(|A^\star|=6\) setting, the six invalid instruments are divided evenly between two nonzero direct-effect values, while the four valid instruments satisfy \(\alpha_j^\star=0\). Because CIIV error was consistently near the oracle 2SLS benchmark when \(|A^\star|=6\) under mixed-sign invalid direct effects and independent instruments, the increased error appears to arise from the interaction between the autoregressive covariance structure and the mixed-sign invalidity pattern.

\subsection{Heterogeneous signed first-stage coefficients with independent instruments}

In this sensitivity analysis, all features of the main simulation design are retained except for the distribution of the first-stage coefficients. Instead of generating positive coefficients from \(\operatorname{Unif}(0.08,0.14)\), we generate heterogeneous signed coefficients according to
\[
\gamma_j^\star
=
S_j A_j,
\qquad
S_j\in\{-1,1\},
\qquad
A_j\sim \operatorname{Unif}(0.03,0.19),
\]
with signs assigned independently across instruments. This design relaxes the positive, similarly sized first-stage coefficient structure used in the main simulations.

\begin{figure}[h!]
\begin{center}
\includegraphics[width=\textwidth]{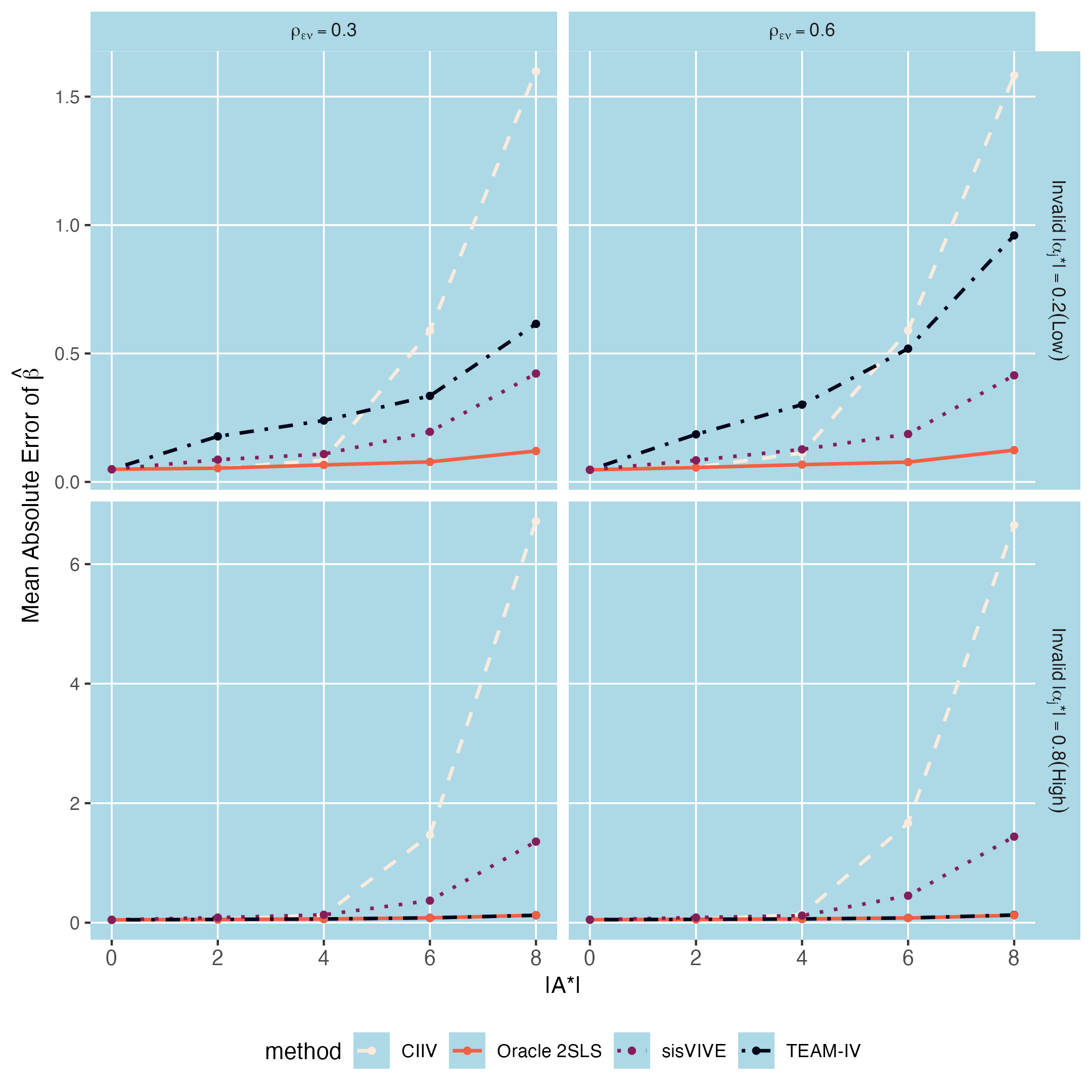}
\caption{Mean absolute error (MAE) of $\widehat\beta$ under the simulation design with independent instrument covariance, positive invalid direct effects, and heterogeneous, mixed-sign first-stage coefficients.  Results are shown by the number of invalid instruments \(|A^\star|\), the endogeneity parameter \(\rho_{\epsilon\nu}\), and the invalid direct-effect magnitude \(a_\alpha\). Each point represents the average across 1000 Monte Carlo replications.}
\label{fig:heterogam-ind-mae}
\end{center}
\end{figure}

Figure~\ref{fig:heterogam-ind-mae} reports the mean absolute estimation error of \(\widehat\beta\) when the first-stage coefficient vector \(\boldsymbol\gamma^\star\) is drawn from a more spread-out uniform probability distribution than in the main simulation design and is allowed to contain both positive and negative coefficients.

The aim of this sensitivity analysis is to determine whether TEAM-IV performance degrades when the first-stage coefficients \(\gamma_j^\star\) are mixed in sign and more widely distributed than in the main simulation design.

Overall, TEAM-IV appears robust to this more heterogeneous, mixed-sign \(\boldsymbol\gamma^\star\). The more noticeable difference is for sisVIVE, which exhibits a more gradual increase in error when the majority rule is violated. CIIV also has lower error than in the main design for the \(|A^\star|=6\) cases.

When \(a_\alpha=0.2\), TEAM-IV performs mildly worse than sisVIVE, and sisVIVE in turn performs mildly worse than oracle 2SLS, with these gaps appearing to grow with the size of the invalid set \(|A^\star|\). CIIV, in contrast, performs almost as well as oracle 2SLS when \(|A^\star|\in\{0,2,4\}\), but otherwise degrades strongly when invalid instruments are in the majority.

When \(a_\alpha=0.8\), sisVIVE continues to perform mildly worse than oracle 2SLS, as under \(a_\alpha=0.2\), but TEAM-IV attains performance nearly identical to oracle 2SLS. CIIV error again increases markedly when \(|A^\star|\in\{6,8\}\).

Overall, this design shows that TEAM-IV is robust to heterogeneity and sign incoherence in the first-stage coefficients \(\boldsymbol\gamma^\star\) under independent instruments. Compared with the main design, sisVIVE performs better under this heterogeneous, mixed-sign first-stage specification, while CIIV shows a similar pattern of degradation when invalid instruments constitute a majority.

\subsection{Heterogeneous signed first-stage coefficients with autoregressive instrument correlation}

\begin{figure}[h!]
\begin{center}
\includegraphics[width=\textwidth]{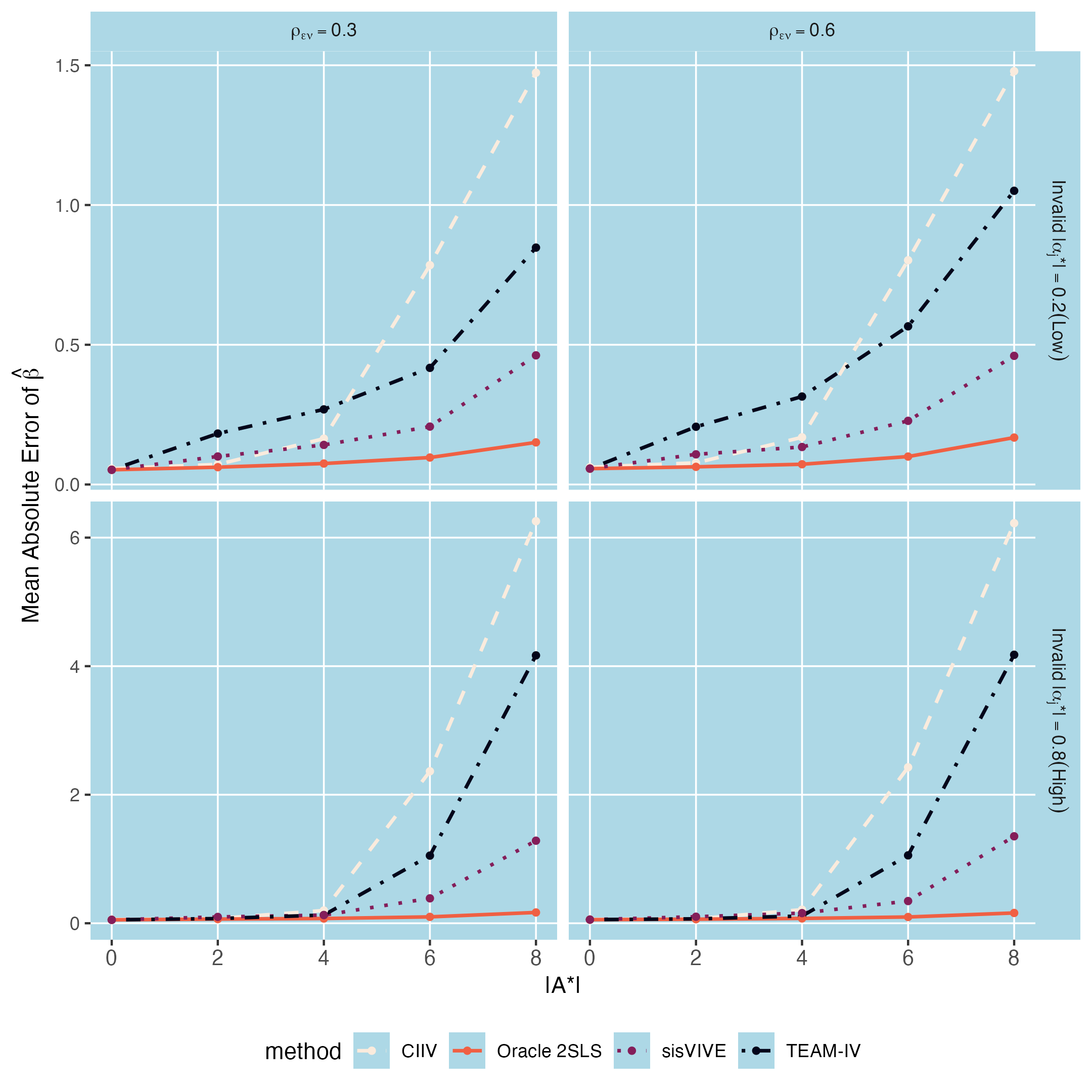}
\caption{Mean absolute estimation error of \(\widehat\beta\) under the simulation design with autoregressive instrument covariance, positive invalid direct effects, and heterogeneous, mixed-sign first-stage coefficients. Results are shown by the number of invalid instruments \(|A^\star|\), the endogeneity parameter \(\rho_{\epsilon\nu}\), and the invalid direct-effect magnitude \(a_\alpha\). A total of 1000 Monte Carlo replications were attempted for each setting; each point represents the average across 980--1000 usable replications, depending on the setting.}
\label{fig:heterogam-ar1-mae}
\end{center}
\end{figure}

Figure~\ref{fig:heterogam-ar1-mae} reports the mean absolute estimation error of \(\widehat\beta\) when the instruments have an autoregressive correlation structure and each first-stage coefficient satisfies
\[
\gamma_j^\star=S_jA_j,
\]
where
\[
S_j\overset{\mathrm{iid}}{\sim}\operatorname{Unif}\{-1,1\},
\qquad
A_j\overset{\mathrm{iid}}{\sim}\operatorname{Uniform}(0.03,0.19),
\qquad
j=1,\ldots,L,
\]
with \(S_j\) and \(A_j\) generated independently.

In a small number of Monte Carlo replications, TEAM-IV did not produce an estimate because the team-construction step returned no team containing more than one instrument. This occurred once when \(|A^\star|=4\), \(\rho_{\epsilon\nu}=0.3\), and \(a_\alpha=0.8\); 15 times when \(|A^\star|=6\), \(\rho_{\epsilon\nu}=0.3\), and \(a_\alpha=0.8\); 10 times when \(|A^\star|=6\), \(\rho_{\epsilon\nu}=0.6\), and \(a_\alpha=0.8\); 20 times when \(|A^\star|=8\), \(\rho_{\epsilon\nu}=0.3\), and \(a_\alpha=0.8\); and 19 times when \(|A^\star|=8\), \(\rho_{\epsilon\nu}=0.6\), and \(a_\alpha=0.8\). To preserve a common comparison set across methods, all methods were summarized in Figure~\ref{fig:heterogam-ar1-mae} using only replications for which TEAM-IV produced an estimate.

The aim of this crossed sensitivity analysis is to determine whether TEAM-IV performance degrades when the first-stage coefficients \(\gamma_j^\star\) are mixed in sign and more widely distributed while the candidate instruments have an autoregressive correlation structure. Overall, TEAM-IV performance does decline in this setting relative to sisVIVE and oracle 2SLS.

When \(a_\alpha=0.2\), TEAM-IV error exceeds that of sisVIVE and oracle 2SLS for all nonzero levels of \(|A^\star|\). When \(a_\alpha=0.6\), TEAM-IV performs similarly to sisVIVE when \(|A^\star|\in\{0,2,4\}\), but degrades relative to sisVIVE when \(|A^\star|\in\{6,8\}\).

CIIV appears to perform close to oracle 2SLS when \(|A^\star|\in\{0,2,4\}\), but its error increases rapidly as \(|A^\star|\) increases beyond 4.

Overall, this sensitivity analysis reveals a limitation of TEAM-IV. When the instruments possess an autoregressive correlation structure and the first-stage coefficients are heterogeneous and mixed in sign, TEAM-IV performs no better than sisVIVE in several settings and performs worse when invalid instruments constitute a majority.

\section*{Acknowledgments}
The author used OpenAI’s ChatGPT to assist with code development and editorial suggestions; all statistical methods, simulations, and conclusions are the author’s own. 

\bibliographystyle{plainnat}

 \end{document}